\documentclass[12pt,letterpaper]{article}

\usepackage[margin=1in]{geometry}

\usepackage{amsfonts, amsmath, amssymb, amsthm}
\usepackage{mathrsfs} 

\usepackage[dvipsnames]{xcolor}
\usepackage{hyperref}
\hypersetup{
    colorlinks=true,
    linkcolor=RedOrange,
    citecolor=MidnightBlue,
    urlcolor=MidnightBlue,
}
\usepackage{graphicx}           
\usepackage{algorithm}
\usepackage{algorithmic}
\usepackage{authblk}
\usepackage{subfig}
\usepackage{physics}

\usepackage[citestyle=alphabetic,bibstyle=authoryear]{biblatex}
\theoremstyle{plain}
\newtheorem{theorem}{Theorem} 
\newtheorem{lemma}{Lemma} 
\newtheorem{proposition}{Proposition}

\newtheorem{assumption}{Assumption}
\theoremstyle{remark}
\newtheorem{definition}{Definition}	
\newtheorem{remark}{Remark}

\newcommand{\cB}{\mathcal{B}}
\newcommand{\cC}{\mathcal{C}}

\newcommand{\cH}{\mathcal{H}}

\newcommand{\cT}{\mathcal{T}}

\newcommand{\cV}{\mathcal{V}}

\newcommand{\bX}{\mathbf{X}}

\newcommand{\bZ}{\mathbf{Z}}

\newcommand{\N}{\mathbb{N}}
\newcommand{\R}{\mathbb{R}}

\DeclareMathOperator*{\E}{\mathbb{E}} 
\renewcommand{\P}{\mathbb{P}}

\DeclareMathOperator*{\argmin}{arg\,min} 

\renewcommand{\epsilon}{\varepsilon}
\newcommand{\indep}{\mathrel{\perp\!\!\!\perp}}

\providecommand{\keywords}[1]{\small \textbf{\textit{Keywords---}} #1}

\title{A Convexified Matching Approach to Imputation and Individualized Inference}
\author[1]{YoonHaeng Hur}
\author[1]{Tengyuan Liang\thanks{\tt Liang acknowledges the generous support from the NSF Career Grant (DMS-2042473), and the William Ladany Faculty Fellowship from the University of Chicago Booth School of Business.}}
\affil[1]{University of Chicago}

\begin{document}

\maketitle

\begin{abstract}
	We introduce a new convexified matching method for missing value imputation and individualized inference inspired by computational optimal transport. Our method integrates favorable features from mainstream imputation approaches: optimal matching, regression imputation, and synthetic control. We impute counterfactual outcomes based on convex combinations of observed outcomes, defined based on an optimal coupling between the treated and control data sets. The optimal coupling problem is considered a convex relaxation to the combinatorial optimal matching problem. We estimate granular-level individual treatment effects while maintaining a desirable aggregate-level summary by properly constraining the coupling. We construct transparent, individual confidence intervals for the estimated counterfactual outcomes. We devise fast iterative entropic-regularized algorithms to solve the optimal coupling problem that scales favorably when the number of units to match is large. Entropic regularization plays a crucial role in both inference and computation; it helps control the width of the individual confidence intervals and design fast optimization algorithms.
\end{abstract}

\keywords{Missing value imputation, optimal coupling, matching, propensity score, individualization, entropic regularization.}


\section{Introduction}
\label{sec-introduction}
One central topic in applied econometric research is to assess the effects of policy interventions reliably. On the one hand, to analyze nonexperimental data, advanced econometric estimates are devised to provide granular counterfactual answers by leveraging specific structural models. In an influential paper in 1986 \cite{lalonde1986evaluating}, LaLonde questioned whether such sophisticated econometric estimates are credible. He compared them with the experimental benchmarks on some coarse summary---for example, the average treatment effect (ATE)---and concluded unfavorably. On the other hand, when experimental data, as in randomized controlled trials, are available, elementary statistical estimates can determine the ATE and eliminate confounding explanations. The ATE estimate, however, does not answer whether the treatment works for an individual. In modern applications such as personalized medicine and online marketing, the treatment effects vary across individuals; the treatment might be beneficial for some individuals but ineffective for others \cite{liu2016there}. One potentially costly approach is to conduct individualized experiments in different time windows, known as N-of-1 trials \cite{hill1961principles,liang2023randomization}. Therefore, it is desirable to develop nonexperimental methods that conform to the coarse estimates---for example, those for ATE or the average treatment effect on the treated (ATT)---at the aggregate level, while delivering more granular, individualized inference.

Individualized inference is about imputing missing counterfactual outcomes---the outcomes that would have been observed had the subjects received an opposite treatment---and their uncertainty quantifications. Three popular approaches are available in the literature: matching, regression imputation, and synthetic control. This paper proposes a new convexified matching method by integrating these three approaches. The central quantity we play is a coupling matrix between the treated and control data clouds, which resembles matching and is used to synthesize the counterfactual outcomes as in classic nonparametric regression.

Matching \cite{rubin2006matched,rosenbaum2020modern} is a widely adopted way to impute counterfactual outcomes. For each subject, the counterfactual outcome is estimated by identifying units in the opposite treatment group similar to the subject in covariates and then averaging their outcomes. Due to its simplicity, nearest neighbor matching, which finds subjects closest under a suitable distance between the covariates, is favored in practice. Exact or approximate matching with high-dimensional covariates can be difficult, and thus propensity score matching was proposed as a remedy \cite{rosenbaum1983central}. There are more sophisticated matching methods to improve balance or forbid certain matches, often referred to as the optimal matching \cite{rosenbaum2020modern,zubizarreta2023handbook}; they require solving combinatorial optimization problems by network optimization techniques or mixed integer programming, which can be computationally expensive. However, even analyzing the aggregate ATE estimates resulting from combinatorial optimization problems is highly nontrivial, and these estimates may differ from simple ATE estimates based on propensity score reweighting \cite{rosenbaum1987model,hirano2001estimation,hirano2003efficient}. As in matching, we solve for a coupling matrix---a convex relaxation to the combinatorial constraints---without access to outcomes, separating the design and analysis phases in an observational study. For a large-scale matching problem with $N$ units, we devise entropic-regularized algorithms from optimal transport \cite{villani_2003,peyre2019computational} to solve convexified matching, calling a matrix scaling subroutine with a number of times nearly independent of $N$; see Section~\ref{sec-optimization}. In contrast, in a conventional matching problem, solvability within polynomial time of $N$ may not be feasible. 

Regression imputation is another approach to estimating counterfactual outcomes. The idea is simple: under the ignorability or unconfoundedness assumption \cite{rosenbaum1983central}, the underlying functions that map the covariates to the responses under the treatment and the control are identifiable and estimable based on observational data. It translates a causal inference problem into a regression problem. Parametric or nonparametric regression techniques estimate the conditional response function under the treatment and the control, given the covariates, and in turn find the conditional average treatment effect (CATE), see \cite{hahn1998role,heckman1997matching,heckman2005structural,macurdy2011flexible}. At the aggregate level, semiparametric efficient estimation of ATE leveraging regression imputation techniques is studied in \cite{hahn1998role}. Admittedly, if the quantity of interest is at the aggregate-level---a one-dimensional parameter such as ATE or ATT---the specific nonparametric regression technique matters less, as long as it estimates the function reasonably well. There has been a fruitful line of literature on the use of flexible machine learning methods in estimating CATE \cite{belloni2014inference,athey2016recursive,chernozhukov2017double,farrell2015robust,farrell2021deep,wager2018estimation}, in place of classic Nadaraya-Watson (NW) kernel nonparametric regression. Albeit naive, let us use the NW estimator as an example to illustrate a shared feature by regression imputation and matching: to compute a counterfactual outcome for a treated unit, NW finds local neighbors in the control group and uses convex weights based on the nonparametric kernel to synthesize the outcome. We will adopt this aggregation by convex weights feature into our convexified matching, whereas our coupling weights are optimized globally. Our coupling weights are determined to optimize imputation quality and individual uncertainty quantification; see Section~\ref{sec-inference}. 

Hybrid methods were proposed to combine matching/propensity score weighting \cite{rosenbaum1987model,hirano2001estimation,hirano2003efficient,li2018balancing} with regression imputation techniques to efficiently estimate ATE, notably the doubly robust estimators \cite{laan2003unified,cattaneo2010efficient,farrell2015robust,farrell2021deep,chernozhukov2022locally}. The estimator compares favorably as it is valid when either the propensity score function or the regression function is consistent, robust to bias due to misspecification. Similar in spirit, bias correction using regression to improve the nearest neighbor matching was studied in \cite{abadie2011bias,abadie2006large}. 

From a different vein, synthetic control \cite{abadie2003economic,abadie2010synthetic,abadie2021using} provides a fresh look at counterfactual imputation: a convex combination of control units may synthesize a subject under treatment better than any one of the control unit. Mathematically, a convex combination of data may simultaneously reduce the approximation error or bias, and, at the same time, reduce the variance driven by the idiosyncratic errors. This convex relaxation idea also eases the computation for the optimization problem in matching. We adopt this convex combination idea to approximate the overall covariate information under treatment by convex combinations of the covariate information under control, defined based on a data set-to-data set coupling; see Section~\ref{sec-method}.

We develop a convexified matching method that solves an optimal coupling, which in turn defines convex combinations for missing value imputation. Unlike combinatorial optimization problems, where the optimization variable must be an extreme point of a certain polytope to represent an assignment, we optimize over the whole polytope as the variable represents the weights of the convex combination. Moreover, we add an entropic regularization to the optimization, where the strength of regularization controls the bias and variance tradeoff. The resulting formulation is a smooth convex optimization problem confined to the polytope that can be solved efficiently. Lastly, by properly specifying the constraints, we can pair with propensity score weighting \cite{rosenbaum1987model,hirano2001estimation,hirano2003efficient} estimators. We can guarantee that the aggregate summary of individual treatment effects coincides with the desired ATE or ATT estimates.

We introduce an inference procedure to construct individual confidence intervals for the estimated counterfactual outcomes based on the convexified matching method. Providing a credible confidence interval around the individual treatment effect can help decide whether to adopt the treatment for each individual. The width of the confidence interval is determined by the approximation error in covariate balancing and the entropy of the convex weights, which comprise the objective function of the optimization. Therefore, how we formulate the convexified matching is directly targeted at individualized inference. Entropic regularization plays a crucial role in inference, namely, controlling the width of the individual confidence intervals, and in computation, namely, designing fast iterative algorithms to solve the optimization.

In summary, we integrate important features from matching, regression imputation, and synthetic control. We impute counterfactual outcomes by convex combinations defined based on an optimal coupling. The coupling is a convex relaxation to optimal matching and can be solved efficiently using iterative matrix scaling subroutines called the Sinkhorn algorithm \cite{sinkhorn1967diagonal,cuturi2013sinkhorn}. We estimate granular individual treatment effects while maintaining a desirable aggregate-level summary by properly constraining the coupling. We construct transparent, individual confidence intervals for the estimated counterfactual outcomes, where the optimization objective controls the width of the confidence intervals.

\paragraph{Notation} We denote by $\|\cdot\|_2$ and $\langle \cdot, \cdot \rangle$ the Euclidean norm and inner product, respectively. For a square matrix $A$, let $\mathrm{tr}(A)$ denote its trace. For matrices $A, B$ of the same dimension, $\langle A, B \rangle:= \mathrm{tr}(A^\top B)$ is the Frobenius inner product. For any integer $n \in \N$, let $1_n = (1, \ldots, 1) \in \R^n$ denote the vector whose entries are all $1$. Let $\Delta_n := \{a \in \R_+^n : \sum_{i = 1}^{n} a_i = 1\}$ denote the simplex of probability vectors and let $\Delta_n^+ := \{a \in \Delta_n : a_i > 0 ~~ \forall i\}$ denote the interior of $\Delta_n$.

\section{Imputation Method: Synthetic Coupling}
\label{sec-method}
We first introduce Kernel Synthetic Coupling (KSC), a convexified matching method for missing value imputation. Later in Section~\ref{sec-inference}, we build individual confidence intervals around the imputed values. To streamline the exposition, we cast the KSC method following the notations in the potential outcome framework \cite{neyman1923,rubin1974estimating} for a binary treatment---a leading example for missing value imputation. 

We consider $N$ units indexed by $1, \ldots, N$, where each unit is associated with two potential outcome variables $Y_i(1), Y_i(0) \in \R$ under treatment and control, respectively, and a covariate vector $x_i \in \R^d$. We can observe only one of $Y_i(1), Y_i(0)$ depending on the treatment assignment $Z_i = 1$ or $Z_i = 0$, denoting that unit $i$ is treated or not. Accordingly, the observed outcome $Y_i$ of unit $i$ is given as $Y_i = Z_i Y_i(1) + (1 - Z_i) Y_i(0)$. We let $\cT = \{i \in [N] : Z_i= 1\}$ and $\cC = \{i \in [N] : Z_i= 0\}$ denote the sets of indices of the treated units and control units, respectively, so that $\cT \cup \cC = \{1, \ldots, N\}$. Also, let $N_t:= |\cT|$ and $N_c:= |\cC|$ denote the numbers of the treated and control units. The individual treatment effect of unit $i$ is defined as $\tau_i = Y_i(1) - Y_i(0)$, which is unobservable as only one of $Y_i(1), Y_i(0)$ can be observed.

\subsection{A Simple Formulation}
\label{sec:main_formulation}
Without loss of generality, we focus on imputing the missing potential outcomes of the treated units; counterfactual outcomes of the control units can be similarly obtained. We estimate $Y_j(0)$ by some $\widehat{Y}_j(0)$ for each treated unit $j$. This, in turn, allows us to estimate the individual treatment effect $\tau_j$ by $Y_j - \widehat{Y}_j(0) = Y_j(1) - \widehat{Y}_j(0)$. The proposed convexified matching combines two ideas. First, we want to find a match between the data cloud in the treated group and that in the control using the covariate information. Second, as in the synthetic control method, the matching quality is determined by how well it approximates the covariate of each treated unit by a convex combination of the covariates of control units. As a result, this procedure involves finding a matrix, which we call a \textit{coupling}, indexed by a pair $(i, j)$ for $i \in \cC$ and $j \in \cT$ denoting the weight assigned to the covariate of control unit $i$ to approximate the covariate of treated unit $j$. We solve the following optimization to obtain the \textit{optimal coupling}:
\begin{align}
	\min_{\pi \in \R_+^{N_c \times N_t}} \quad & \frac{1}{2 N_t} \sum_{j \in \cT} \|x_j -  \sum_{i \in \cC} \tfrac{\pi_{i j}}{1 / N_t} x_i \|_2^2 + \lambda \sum_{i \in \cC} \sum_{j \in \cT} \pi_{i j} \log \tfrac{\pi_{i j}}{e}, \label{eq:objective} \\
	\mathrm{subject~to} \quad & \sum_{i \in \cC} \pi_{i j} = 1 / N_t \quad \forall j \in \cT, \label{eq:weights_sum_to_one} \\
	& \sum_{j \in \cT} \pi_{i j} = 1 / N_c \quad \forall i \in \cC. \label{eq:weights_constraint_mean-difference}
\end{align}

The first term in the objective function \eqref{eq:objective} is the average of the squared approximation errors for covariate balancing. In the first constraint \eqref{eq:weights_sum_to_one}, for each treated unit $j$, the synthetic weights $(\tfrac{\pi_{i j}}{1 / N_t})_{i \in \cC}$ sum to $1$ and thus the term $\sum_{i \in \cC} \tfrac{\pi_{i j}}{1 / N_t} x_i$ in \eqref{eq:objective} is a convex combination to approximate $x_j$. 

The second term in \eqref{eq:objective} is an entropic regularization with regularization parameter $\lambda > 0$. Increasing the parameter $\lambda$ encourages the solution $\pi$ to be uniform, namely, $\pi_{i j} = \tfrac{1}{N_c \cdot N_t}$ for all $(i, j)$'s, which essentially corresponds to increasing the number of neighbors in the nearest neighbor matching. 

The second constraint \eqref{eq:weights_constraint_mean-difference} enforces that all control units contribute equally and thus seeks a matching between treated and control data sets; this constraint can be generalized to incorporate inverse propensity scores, to be shown in the next section. If this constraint is dropped, the optimization can be decoupled into $N_t$ independent optimization of the same type, each with one treated unit $j \in \cT$, also solvable by our optimization algorithm detailed in Section~\ref{sec-optimization}. We shall show in a second that this constraint will be crucial. It enables the convex program to synthesize granular information while enforcing, at the coarse level, that the answers agree with typical estimators for the ATE or ATT. 

Finally, the optimization variable $\pi \in \R_+^{N_c \times N_t}$ is a matrix whose entries are indexed by $(i, j)$ for $i \in \cC$ and $j \in \cT$, instead of the usual indexing by $i \in \{1, \ldots, N_c\}$ and $j \in \{1, \ldots, N_t\}$, to denote that $\pi_{i j}$ is the weight between control unit $i$ and treated unit $j$. We call $\pi \in \R_+^{N_c \times N_t}$ satisfying the constraints \eqref{eq:weights_sum_to_one} and \eqref{eq:weights_constraint_mean-difference} a \textit{coupling}, which generalizes doubly stochastic matrices to non-square matrices with prescribed row and column sums.

Upon obtaining the solution, denoted by $\widehat{\pi}$, we impute the counterfactual outcome $Y_j(0)$ of treated unit $j$ by the convex combination of the outcomes of control units with the weights $(N_t \widehat{\pi}_{i j})_{i \in \cC}$:
\begin{equation}
	\label{eqn:convex-combination-imputation}
	\widehat{Y}_j(0) := \sum_{i \in \cC}  \tfrac{\widehat{\pi}_{i j} }{ 1/N_t } Y_i,
\end{equation}
which in turn leads to the individual treatment effect estimate $\widehat{\tau}_j := Y_j - \widehat{Y}_j(0)$. By the constraint \eqref{eq:weights_constraint_mean-difference}, the average of the estimated individual treatment effects coincides with the difference in the mean outcomes between the treatment and control groups denoted as $\widehat{\tau}^{\mathsf{DiM}}$:
\begin{equation}
	\label{eq:mean_difference}
	\frac{1}{N_t} \sum_{j \in \cT} \widehat{\tau}_j = \frac{1}{N_t} \sum_{j \in \cT} (Y_j - \sum_{i \in \cC} \tfrac{\widehat{\pi}_{i j} }{ 1/N_t } Y_i) = \frac{1}{N_t} \sum_{j \in \cT} Y_j - \frac{1}{N_c} \sum_{i \in \cC} Y_i =: \widehat{\tau}^{\mathsf{DiM}}.
\end{equation}
Therefore, the proposed convexified matching allows for estimating individual treatment effects while maintaining the desired ATT estimate $\widehat{\tau}^{\mathsf{DiM}}$.

The resulting optimization is a smooth convex optimization problem. To see this, notice that the first term of the objective function \eqref{eq:objective} is the following quadratic function of $\pi$:
\begin{equation}
	\label{eq:objective_quadratic}
	\frac{N_t}{2} \langle \pi, K_{c c} \pi \rangle - \langle \pi, K_{c t} \rangle + \frac{\mathrm{tr}(K_{t t})}{2 N_t},
\end{equation}
where $K_{c c} \in \R^{N_c \times N_c}$, $K_{c t} \in \R^{N_c \times N_t}$, and $K_{t t} \in \R^{N_t \times N_t}$ are the Gram matrices whose entries are inner products of the covariates, namely, $K_{c c} = (\langle x_i, x_{i'} \rangle)_{i, i' \in \cC}$, $K_{c t} = (\langle x_i, x_j \rangle)_{(i, j) \in \cC \times \cT}$, and $K_{c, t} = (\langle x_j, x_{j'} \rangle)_{j, j' \in \cT}$. This quadratic function is convex as the Gram matrix $K_{c c}$ is positive semidefinite. The second term---the entropy term---is convex. The constraint set resulting from \eqref{eq:weights_sum_to_one} and \eqref{eq:weights_constraint_mean-difference} is a convex polytope consisting of couplings. In Section \ref{sec-optimization}, we introduce and analyze efficient algorithms to solve this optimization problem based on a simple iterative matrix scaling procedure called the Sinkhorn algorithm \cite{sinkhorn1967diagonal}. 

Lastly, we discuss several existing methods and concepts related to the proposed method. \cite{abadie2021penalized} considers multiple treated units ($N_t > 1$) and proposes running synthetic controls separately for treated units with a different regularization function. The proposed method differs from running synthetic controls separately because of the constraint \eqref{eq:weights_constraint_mean-difference}. If there is only one treated unit ($N_t = 1$) and no constraint \eqref{eq:weights_constraint_mean-difference}, the proposed method is equivalent to the standard synthetic control with entropic regularization as in \cite{hainmueller2012entropy} which proposes entropic regularization for covariate balancing. Also, notice that the first term of \eqref{eq:objective} is upper bounded under the constraint \eqref{eq:weights_sum_to_one} as follows: 
\begin{equation*}
	\frac{1}{2 N_t} \sum_{j \in \cT} \|x_j - \sum_{i \in \cC} \tfrac{\pi_{i j}}{1 / N_t} x_i \|_2^2 
	= \frac{1}{2 N_t} \sum_{j \in \cT} \| \sum_{i \in \cC} \tfrac{\pi_{i j}}{1 / N_t} (x_j - x_i)\|_2^2 
	\le \frac{1}{2 N_t} \sum_{j \in \cT} \sum_{i \in \cC} \tfrac{\pi_{i j} }{1 / N_t} \|x_j - x_i\|_2^2,
\end{equation*}
where the inequality uses Jensen's inequality. This upper bound analytically shows that the average approximation error based on convex combinations is smaller than that based on pairwise distances, which aligns with the motivation mentioned in the previous section. Replacing the first term of the objective function \eqref{eq:objective} by the above upper bound---while maintaining the constraints \eqref{eq:weights_sum_to_one}, \eqref{eq:weights_constraint_mean-difference}---leads to the optimal transport problem \cite{villani_2003} with entropic regularization \cite{cuturi2013sinkhorn}.

\subsection{Extensions}
\label{sec:extensions}
We present two extensions of the simple formulation introduced in Section \ref{sec:main_formulation}. The first is extending the main insights to cover the nonlinear case when the potential outcomes can be nonlinear functions of covariates, using a kernel trick \cite{shawe2004kernel,steinwart_christmann_2008}. The second is incorporating notions of propensity scores \cite{rosenbaum1983central} into our synthetic coupling formulation.

\paragraph{Kernel Synthetic Coupling}

We can kernelize the presented simple formulation to extend to nonlinear cases. The kernel trick is widely used in machine learning to modify methods that originated in the linear setting to accommodate nonlinearity, such as ridge regression, support vector machines, and principal component analysis. Recall that our objective \eqref{eq:objective} relies on the Gram matrices as shown in \eqref{eq:objective_quadratic}. Therefore, to kernelize the simple formulation, we can replace the Gram matrices with the kernel matrices, namely, letting $K_{c c} = (k(x_i, x_{i'}))_{i, i' \in \cC}$, $K_{c t} = (k(x_i, x_j))_{(i, j) \in \cC \times \cT}$, and $K_{t t} = (k(x_j, x_{j'}))_{j, j' \in \cT}$ for some kernel function $k \colon \R^d \times \R^d \to \R$. The kernelized formulation is essentially the same convex optimization problem. It can be tackled by the algorithms we introduce in Section~\ref{sec-optimization}. 

Applying the above kernel trick is equivalent to replacing the first term of \eqref{eq:objective} with the approximation error in the reproducing kernel Hilbert space (RKHS) $\cH$ associated with the kernel $k$, namely, $\frac{1}{2 N_t} \sum_{j \in \cT} \|\phi_{x_j} - \sum_{i \in \cC} \tfrac{\pi_{i j}}{1 / N_t} \phi_{x_i}\|_\cH^2$, where $\|\cdot\|_\cH$ is the norm in the RKHS $\cH$ and $\phi_{x} = k(x, \cdot)$ denotes the canonical feature of $x \in \R^d$. In other words, we generalize the matching of covariates in the Euclidean space to the matching of embedded canonical features in the RKHS $\cH$, which leads to the following kernelized formulation:
\begin{equation}
	\label{eq:kernelized_formulation}
	\begin{aligned}
		\mathrm{KSC}^{v, w}_\lambda := \min_{\pi \in \R_+^{N_c \times N_t}} \quad & \frac{1}{2} \sum_{j \in \cT} v_j \|\phi_{x_j} -  \sum_{i \in \cC} \tfrac{\pi_{i j}}{v_j} \phi_{x_i}\|_\cH^2 + \lambda \sum_{i \in \cC} \sum_{j \in \cT} \pi_{i j} \log \tfrac{\pi_{i j}}{e}, \\
		\mathrm{subject~to} \quad & \sum_{i \in \cC} \pi_{i j} = v_j \quad \forall j \in \cT \quad \text{and} \quad \sum_{j \in \cT} \pi_{i j} = w_i \quad \forall i \in \cC,
	\end{aligned}
\end{equation} 
where $v = (v_j)_{j \in \cT} \in \Delta_{N_t}^+$ and $w = (w_i)_{i \in \cC} \in \Delta_{N_c}^+$ are strictly positive probability vectors that generalize the constraints \eqref{eq:weights_sum_to_one} and \eqref{eq:weights_constraint_mean-difference} to allow treated and control units to have different weights.\footnote{Whenever $v = \frac{1}{N_t} 1_{N_t}$, we will drop the dependency on $v$ and abbreviate as $\mathrm{KSC}^{w}_\lambda$.} With the optimal coupling $\widehat{\pi}$, we impute using convex combination same as in \eqref{eqn:convex-combination-imputation}, a classic idea in nonparametric regression.

\paragraph{Propensity Score}
Using arbitrary weights through $v = (v_j)_{j \in \cT} \in \Delta_{N_t}^+$ and $w = (w_i)_{i \in \cC} \in \Delta_{N_c}^+$ allows for various average treatment effect estimators, particularly including those incorporating propensity scores. To estimate the ATT, we can set $v = \frac{1}{N_t} 1_{N_t}$ as before, and the average of the estimated individual treatment effects of the treated leads to
\begin{equation*}
 	\frac{1}{N_t} \sum_{j \in \cT} \widehat{\tau}_j = \frac{1}{N_t} \sum_{j \in \cT} (Y_j - \sum_{i \in \cC} \tfrac{\widehat{\pi}_{i j} }{1 / N_t} Y_i) =  \frac{1}{N_t} \sum_{j \in \cT} Y_j - \sum_{i \in \cC} w_i Y_i.
\end{equation*}
Now, let $\widehat{p} \colon \R^d \to (0, 1)$ be a suitable estimator---often estimated by logistic regression---of the propensity score, $x \mapsto \P(W = 1 \, | \, X = x)$. 
If the weights $w$ are chosen based on the propensity scores as
\begin{equation}
	\label{eq:propensity_score_weights}
	w_i = \tfrac{\widehat{p}(x_i)}{1 - \widehat{p}(x_i)} / \sum_{i' \in \cC} \tfrac{\widehat{p}(x_{i'})}{1 - \widehat{p}(x_{i'})} \quad \forall i \in \cC,
\end{equation}
the aggregate effects match the normalized Inverse Probability Weighting (IPW) estimator \cite{hirano2001estimation}, that is,
\begin{equation}
	\label{eq:normalized_ipw}
	\frac{1}{N_t} \sum_{j \in \cT} \widehat{\tau}_j  = \frac{\sum_{j \in \cT} Y_j}{N_t}  - \frac{\sum_{i \in \cC}  \tfrac{\widehat{p}(x_i)}{1 - \widehat{p}(x_i)} Y_i }{\sum_{i \in \cC} \tfrac{\widehat{p}(x_i)}{1 - \widehat{p}(x_i)}} =: \widehat{\tau}^{\mathsf{IPW}}_{\mathsf{ATT}}.
\end{equation}
Namely, the proposed KSC method is guaranteed to yield individual treatment effect estimates that conform to \eqref{eq:normalized_ipw} at the aggregate level. 

Similarly, solving $\mathrm{KSC}^{v,w}_\lambda$ in \eqref{eq:kernelized_formulation} by setting
\begin{align*}
	v_j = \tfrac{1}{\widehat{p}(x_j)} / \sum_{j' \in \cT} \tfrac{1}{\widehat{p}(x_{j'})} \quad \forall j \in \cT \quad \text{and} \quad w_i = \tfrac{1}{1 - \widehat{p}(x_i)} / \sum_{i' \in \cC} \tfrac{1}{1 - \widehat{p}(x_{i'})} \quad \forall i \in \cC,
\end{align*}
we obtain an coupling $\widehat{\pi}$ that matches the IPW estimator for ATE
\begin{align*}
	 \sum_{j \in \cT} v_j \widehat{\tau}_j = \frac{\sum_{j \in \cT} \tfrac{1}{\widehat{p}(x_j)} Y_j}{\sum_{j \in \cT} \tfrac{1}{\widehat{p}(x_j)}} - \frac{\sum_{i \in \cC}  \tfrac{1}{1 - \widehat{p}(x_i)}  Y_i}{\sum_{i \in \cC}  \tfrac{1}{1 - \widehat{p}(x_i)} } =:  \widehat{\tau}^{\mathsf{IPW}}_{\mathsf{ATE}}.
\end{align*}

\subsection{Numeric Illustration}
\label{sec:numerical_illustration}

Now, we demonstrate our imputation method, KSC, using a real-world data set. This section aims to concisely present the features of KSC and compare it with other imputation methods; further details and robust evaluations are provided in Section \ref{sec-application} together with the inference procedure, to be presented in the next section.

We apply the proposed method to evaluate the National Supported Work (NSW) demonstration program, first analyzed by \cite{lalonde1986evaluating}. This data set has been widely studied in the program evaluation literature. The NSW program was conducted from 1975 to 1978 in the United States, which aimed at providing a job training program to disadvantaged workers, where the treatment, namely, the training, was randomly assigned to some of them. The data we consider here is a specific subset of the experimental data used in \cite{dehejia1999causal}, which consists of $N = 445$ subjects with $N_t = 185$ treated units and $N_c = 260$ control units. The outcome of interest $Y_i$ is the post-treatment earnings recorded in 1978. Each individual is associated with six variables: age, years of education, and four indicator variables denoting whether the individual is black, Hispanic, married, and a high school dropout. Following \cite{dehejia1999causal}, the pretreatment outcomes, the earnings in 1974 and 1975, and two indicator variables denoting whether these pretreatment earnings are zero, are included as covariates. Accordingly, we have $x_i \in \R^d$ with $d = 10$. 

\begin{figure}
	\centering
	\subfloat[NSW $\lambda = 0.001$]{\includegraphics*[width=0.47\textwidth]{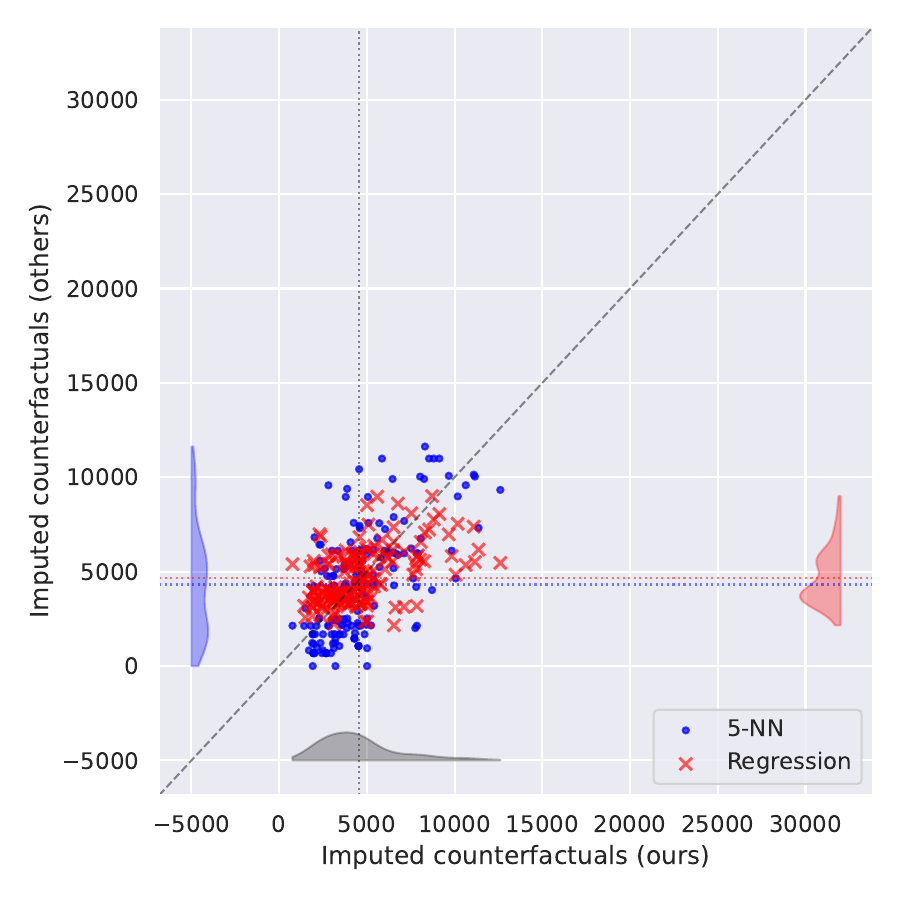}}
	\subfloat[NSW $\lambda = 0.01$]{\includegraphics*[width=0.47\textwidth]{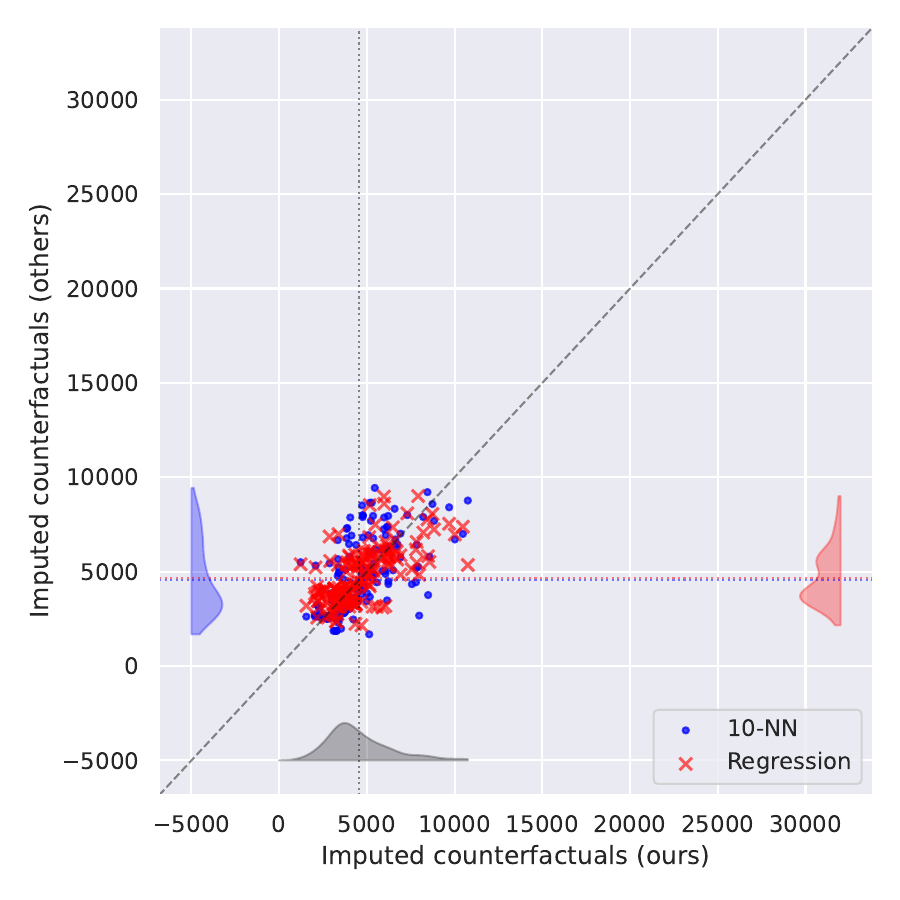}}
	\qquad
	\subfloat[PSID $\lambda = 0.001$]{\includegraphics*[width=0.47\textwidth]{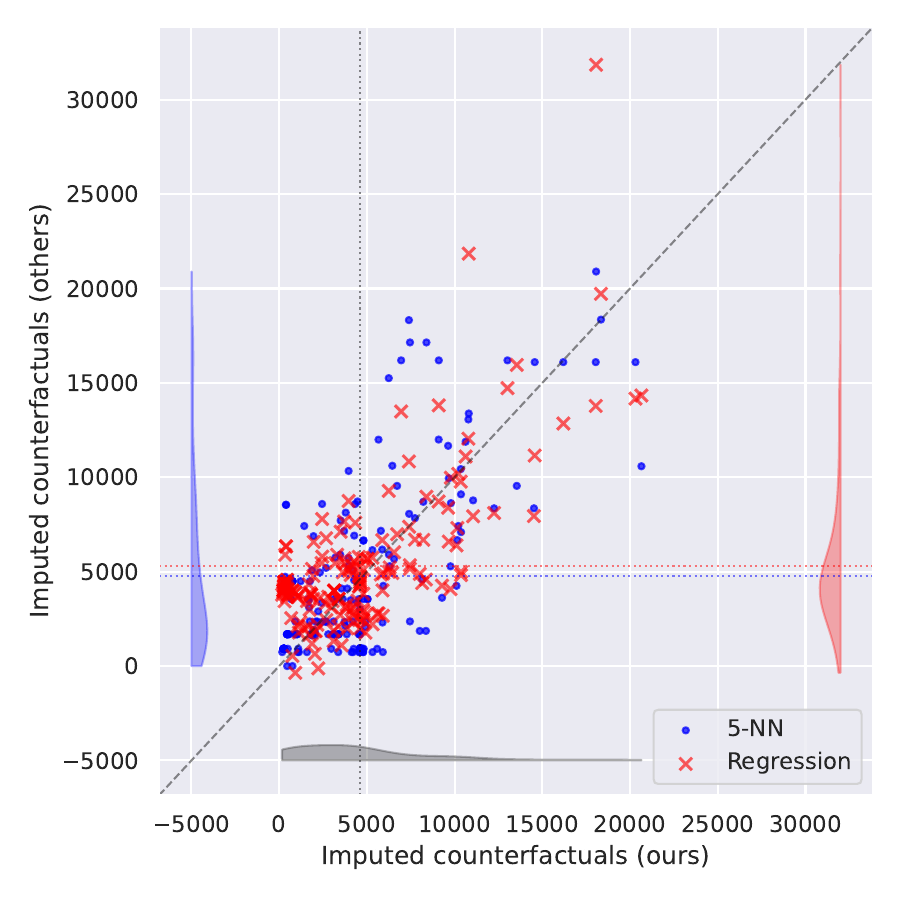}}
	\subfloat[PSID $\lambda = 0.01$]{\includegraphics*[width=0.47\textwidth]{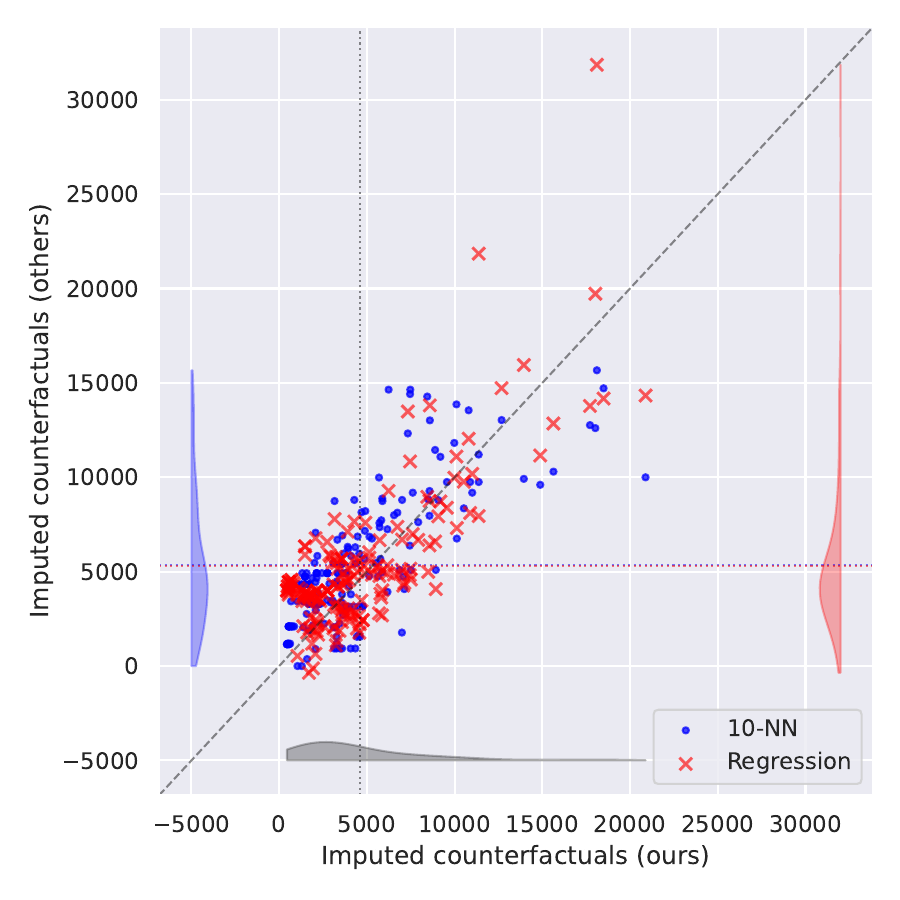}}
	\caption{\small Scatter plots of the imputed counterfactual outcomes of the treated along with the histograms of the marginal distributions. The $x$-coordinate is $\widehat{Y}_j(0)$ imputed by the proposed method with $\lambda \in \{0.001, 0.01\}$, while the $y$-coordinate is the imputed counterfactual outcomes by the nearest neighbor matching with replacement or the regression method. (a) and (b) are based on the NSW experimental data, while (c) and (d) are based on the PSID data after trimming. The proposed method uses the linear kernel for both data. For the NSW data, we use uniform weights $v, w$ in \eqref{eq:kernelized_formulation}, while for the PSID data, we set $w$ to be the propensity score-based weights following \eqref{eq:propensity_score_weights}.}
	\label{fig:NSW_imputation_scatterplots}
\end{figure}

Figure \ref{fig:NSW_imputation_scatterplots}(a) and Figure \ref{fig:NSW_imputation_scatterplots}(b) show the scatter plots of the imputed counterfactual outcomes of the treated units, where the $x$-coordinate is $\widehat{Y}_j(0)$ imputed by the proposed method, while the $y$-coordinate is the imputed counterfactual outcomes by the nearest neighbor matching with replacement or the regression imputation method. The proposed method relies on the linear kernel following the simple formulation in Section \ref{sec:main_formulation} with $\lambda \in \{0.001, 0.01\}$, the nearest neighbor matching uses $5$ and $10$ nearest neighbors, and the regression method fits a linear regression model to the observed outcome using all the covariates and the treatment indicator. At first glance, the scatter plots visually suggest that all three imputation methods provide comparable imputed counterfactual outcomes. One crucial difference is that the average of the imputed counterfactual outcomes and, thus, the estimated individual treatment effects vary across the methods. The ATT estimate computed by the proposed method always matches---regardless of the choice of $\lambda$---the mean difference between the treatment and control groups, which is roughly 1794.3, as shown in \eqref{eq:mean_difference}---the unbiased estimator of the ATT in the experimental setup. In contrast, the $k$-nearest neighbor matching with $k = 5$, with $k = 10$, and the regression imputation method produce the ATT estimates 2030.5, 1776.6, and 1706.2, respectively. While the average of the imputed counterfactual outcomes does not depend on the choice of $\lambda$, we can see that it determines the level of individualization. The histograms along the $x$-axix show that the imputed counterfactual outcomes based on the proposed method with $\lambda = 0.001$ are more dispersed than those based on $\lambda = 0.01$; the latter is based on larger regularization leading to more uniform weights, which is similar to the fact that $k = 5$ is more dispersed---along the $y$-axis---than $k = 10$ in the nearest neighbor matching.

Since \cite{lalonde1986evaluating}, there has been a debate on the credibility of the average treatment effect estimation using observational control groups. Since such control groups can be significantly different from the control group of the experimental data, the resulting ATE estimates can deviate much from the experimental benchmark, say, the mean difference of 1794.3 computed earlier, which has been the central argument of \cite{lalonde1986evaluating}. Later, several methods \cite{dehejia1999causal,dehejia2002propensity} have been proposed to recover the experimental benchmark estimate using observational control groups. Here, we focus on one such method based on the normalized inverse probability weighting (IPW) estimator \eqref{eq:normalized_ipw} and apply the KSC method with the weights based on the propensity scores. To this end, we consider a control group from the Panel Study of Income Dynamics (PSID) data, consisting of 2490 units, with which we estimate the propensity scores by logistic regression. The resulting ATT estimate by $\widehat{\tau}^{\mathsf{IPW}}_{\mathsf{ATT}}$ is 2579.7, which is far from the benchmark of 1794.3, rooted in the stark difference between the experimental and observational control groups. To remedy this, we trim the PSID data by only taking the units whose propensity scores are in $[0.05, 0.95]$, yielding 214 units. The resulting ATT estimate by $\widehat{\tau}^{\mathsf{IPW}}_{\mathsf{ATT}}$ is then 1748.0, which resembles the experimental benchmark. Taking this trimmed PSID data as a control group, we apply the proposed method with the weights $w_i$ based on the propensity scores as in \eqref{eq:propensity_score_weights} and compare with the matching and regression methods. The results are shown in Figure \ref{fig:NSW_imputation_scatterplots}(c) and Figure \ref{fig:NSW_imputation_scatterplots}(d). Again, though visually similar in distributions, the ATT estimates differ across the methods, where the proposed method with the propensity score-based weights provides the ATT estimate 1748.0, regardless of $\lambda$. In contrast, the $k$-nearest neighbor matching with $k = 5$, with $k = 10$, and the regression method produce the ATT estimates 1567.3, 1003.0, and 1065.5, respectively. This example illustrates that the proposed imputation method allows for producing granular information, namely, the individual treatment effects, from a coarse level estimate by $\widehat{\tau}^{\mathsf{IPW}}_{\mathsf{ATT}}$. The coarse level estimate is closer to the experimental benchmark estimate than inherently imputation-based methods such as the nearest neighbor matching or regression.

\section{Inference: Individual Confidence Intervals}
\label{sec-inference}
Constructing confidence intervals for the estimated individual treatment effects is conducive to reliable decision-making. This section introduces and analyzes the method to construct individual confidence intervals based on the proposed convexified matching. To this end, we build a confidence interval for the missing counterfactual outcome $Y_j(0)$ for each treated unit $j \in \cT$. We first introduce appropriate model assumptions. Then, we elucidate the construction of individual confidence intervals, followed by a numerical example, and discuss the coverage and efficiency of the proposed method.

\subsection{Setup}
\paragraph{The Model}
Let $\cH$ be the reproducing kernel Hilbert space (RKHS) associated with a kernel function $k \colon \R^d \times \R^d \to \R$, where $\langle \cdot, \cdot \rangle_\cH$ denotes the RKHS inner product, $\|\cdot\|_\cH$ denotes the RKHS norm, and $\phi_{x}:= k(x, \cdot) \in \cH$ is the corresponding reproducing kernel function. We postulate the following model for the potential outcomes: for each $i = 1, \ldots, N$, letting $x_i \in \R^d$ be the covariate of unit $i$, we have
\begin{align*}
	Y_i(0) & = f_0(x_i) + \epsilon_{0, i}, ~\epsilon_{0, i} \sim \mathcal{N}(0, \sigma_0^2), \\
	Y_i(1) & = f_1(x_i) + \epsilon_{1, i}, ~\epsilon_{1, i} \sim \mathcal{N}(0, \sigma_1^2),
\end{align*}
where $f_0, f_1 \in \cH$ are some functions in the RKHS and $\sigma_0, \sigma_1 \ge 0$. Let $\bX:= [x_1, \ldots, x_N]$ denote the fixed design matrix and let $\bZ := [Z_1, \ldots, Z_N]$ denote the possibly random treatment assignment vector.

\begin{assumption}
	\label{assumption:inference}
	$\{\epsilon_{0, i}, \epsilon_{1, i}\}_{i = 1}^{N}$ are independent and
	\begin{equation}
		\label{eq:exogeneous_error}
		\bZ \indep \{\epsilon_{0, i}, \epsilon_{1, i}\}_{i = 1}^{N} ~ | ~ \bX,
	\end{equation}
	that is, the treatment assignment and the errors are independent given the covariates.
\end{assumption}

\begin{remark}
	Typical assumptions in causal inference involve the following: (i) $(Y_i(0), Y_i(1), X_i)_{i = 1}^{N}$ are independently drawn from some population distribution defined on $\R \times \R \times \R^d$ and (ii) the treatment assignment $Z_i$ is independent of $(Y_i(0), Y_i(1))$ conditional on $X_i$. Unlike (i), our model postulates a signal-plus-noise structure for the potential outcomes, where the covariates are from a fixed design matrix.  All our results are conditioned on the fixed design $\bX$, which we highlight by using the lowercase $\{x_i\}$ notation instead of $\{X_i\}$. On the one hand, Assumption \ref{assumption:inference} serves the same purpose as (ii) in spirit but is stronger in the sense that the errors are assumed to be Gaussian. On the other hand, we do not require $\bX$ nor $\bZ$ to be i.i.d.\ drawn, allowing $Z_i$'s to have arbitrary dependence across units or even to be adversarially chosen conditioned on the fixed design $\bX$.
\end{remark}

Under this setup, for each treated unit $j \in \cT$, the goal is to build an individual confidence interval for its missing counterfactual outcome $Y_j(0)$. More specifically, we aim to construct a confidence interval $[\widehat{L}_j, \widehat{U}_j]$ that contains the conditional expectation $f_0(x_j) = \E[Y_j(0) \, | \, x_j] = \E[Y_j(0) \, | \, \bX, \bZ]$ with a desired level of confidence as follows:
\begin{equation}
	\label{eq:confidence_interval_coverage}
	\P\left(f_0(x_j) \in [\widehat{L}_j, \widehat{U}_j] ~ | ~ \bX, \bZ\right) \ge 1 - \alpha,
\end{equation} 
namely, the probability that the interval $[\widehat{L}_j, \widehat{U}_j]$ contains the conditional expectation $f_0(x_j)$ is at least $1 - \alpha$, where the probability is conditional on the design matrix $\bX$ and the treatment assignments $\bZ$.

\subsection{Individual Confidence Intervals}
\label{sec:confidence_intervals_method}

\paragraph{Construction}
We consider the following Kernerlized Synthetic Coupling (KSC) discussed in Section \ref{sec:extensions}.
To construct an individual interval for the counterfactual control outcome of each treated unit, we solve \eqref{eq:kernelized_formulation} with uniform weights $v = \frac{1}{N_t}$ on the treated, which we rewrite in the matrix form as follows: let $(w_i)_{i \in \cC}$ be positive numbers such that $\sum_{i \in \cC} w_i = 1$,
\begin{equation}
	\label{eq:kernelized_formulation_uniform_treated}
	\begin{aligned}
		\mathrm{KSC}^{w}_{\lambda} = \min_{\pi \in \R_+^{N_c \times N_t}} \quad & \tfrac{N_t}{2} \langle \pi, K_{c c} \pi \rangle - \langle \pi, K_{c t} \rangle + \tfrac{\mathrm{tr}(K_{t t})}{2 N_t} + \lambda \langle \pi, \log \tfrac{\pi}{e} \rangle, \\
		\mathrm{subject~to} \quad & \sum_{i \in \cC} \pi_{i j} = \tfrac{1}{N_t} \quad \forall j \in \cT \quad \text{and} \quad \sum_{j \in \cT} \pi_{i j} = w_i \quad \forall i \in \cC.
	\end{aligned}
\end{equation}

Let $\widehat{\pi}$ be the solution of the above optimization problem and define a probability transition matrix $\widehat{P} = (\widehat{p}_{i j})_{(i, j) \in \cC \times \cT}$, where $\widehat{p}_{i j} = N_t \widehat{\pi}_{i j}$ so that $\sum_{i \in \cC} \widehat{p}_{i j} = 1$ for any $j \in \cT$. For each treated unit $j$, recall the imputed value $\widehat{Y}_j(0)$ defined in \eqref{eqn:convex-combination-imputation}, for which we construct a confidence interval $[\widehat{L}_j, \widehat{U}_j]$, where $\widehat{L}_j$ and $\widehat{U}_j$ are defined as follows:
\begin{align}
	\widehat{L}_j & := \widehat{Y}_j(0) - \widehat{\theta} \cdot \sqrt{\left(K_{t t} + \widehat{P}^\top K_{c c} \widehat{P} - 2 K_{c t}^\top \widehat{P} \right)_{j j}} - z_{1 - \frac{\alpha}{2}} \cdot \widehat{\sigma}_0 \sqrt{\left(\widehat{P}^\top \widehat{P}\right)_{j j}}, \label{eq:lower} \\
	\widehat{U}_j & := \widehat{Y}_j(0) + \widehat{\theta} \cdot \sqrt{\left(K_{t t} + \widehat{P}^\top K_{c c} \widehat{P} - 2 K_{c t}^\top \widehat{P} \right)_{j j}} + z_{1 - \frac{\alpha}{2}} \cdot \widehat{\sigma}_0 \sqrt{\left(\widehat{P}^\top \widehat{P}\right)_{j j}}. \label{eq:upper}
\end{align}
Here, $z_{1 - \frac{\alpha}{2}}$ is the $(1 - \frac{\alpha}{2})$-quantile of the standard normal distribution and $\widehat{\theta}, \widehat{\sigma}_0$ are suitable estimates of the norm $\|f_0\|_\cH$ and the standard deviation $\sigma_0$, respectively. 

We apply kernel ridge regression to the control group data to estimate $\|f_0\|_\cH$ and $\sigma_0$. We can estimate $f_0 \in \cH$ by $\widehat{f}_0(x) = \sum_{i \in \cC} \widehat{\beta}_i k(x_i, x)$, where $\widehat{\beta} = (K_{c c} + \rho I_{N_c})^{-1} Y_c \in \R^{N_c}$ with a kernel matrix $K_{c c}$ defined in Section \ref{sec:extensions}, a vector of control outcomes $Y_c = (Y_i)_{i \in \cC} \in \R^{N_c}$, and a ridge regularization parameter $\rho > 0$. Then, we estimate $\|f_0\|_\cH$ by
\begin{equation*}
	\widehat{\theta} := \|\widehat{f}_0\|_\cH = \sqrt{\langle \widehat{\beta}, K_{c c} \widehat{\beta} \rangle},
\end{equation*}
and the variance $\sigma_0^2$ by
\begin{equation*}
	\widehat{\sigma}_0^2 = \frac{1}{N_c} \sum_{i \in \cC} (Y_i - \widehat{f}_0(x_i))^2 = \frac{1}{N_c} \langle Y_c - K_{c c} \widehat{\beta}, Y_c - K_{c c} \widehat{\beta} \rangle.
\end{equation*}

Before diving into each component of the construction, we first look at a simple numerical example to illustrate the performance of the proposed method.

\paragraph{Numerical Example}
We consider a set of fixed one-dimensional covariates $x_1, \ldots, x_N \in [0, 1]$ with $N = 500$ and we randomly choose $N_t = 200$ treatment units. Let $\cH$ be the RKHS associated with the Gaussian kernel $k(x, x') = \exp(-\gamma \cdot |x - x'|^2)$, where $\gamma = 2.5$, and suppose $f_0(x) = k(0.5, x)$, which yields $\|f_0\|_\cH = (k(0.5, 0.5))^{1 / 2} = 1$. Then, we generate the potential outcomes by $Y_i(0) = f_0(x_i) + \epsilon_{0, i}$, where $\epsilon_{0, 1}, \ldots, \epsilon_{0, N}$ are independently drawn from $N(0, \sigma_0^2)$. As $x_1, \ldots, x_N$ and $Z_1, \ldots, Z_N$ are fixed, the randomness is solely from the residuals $\epsilon_i$'s, which we sample 1000 times. 

\begin{figure}[!ht]
    \centering
    \subfloat[$\sigma_0 = 0.1$]{\includegraphics[width=0.99\textwidth]{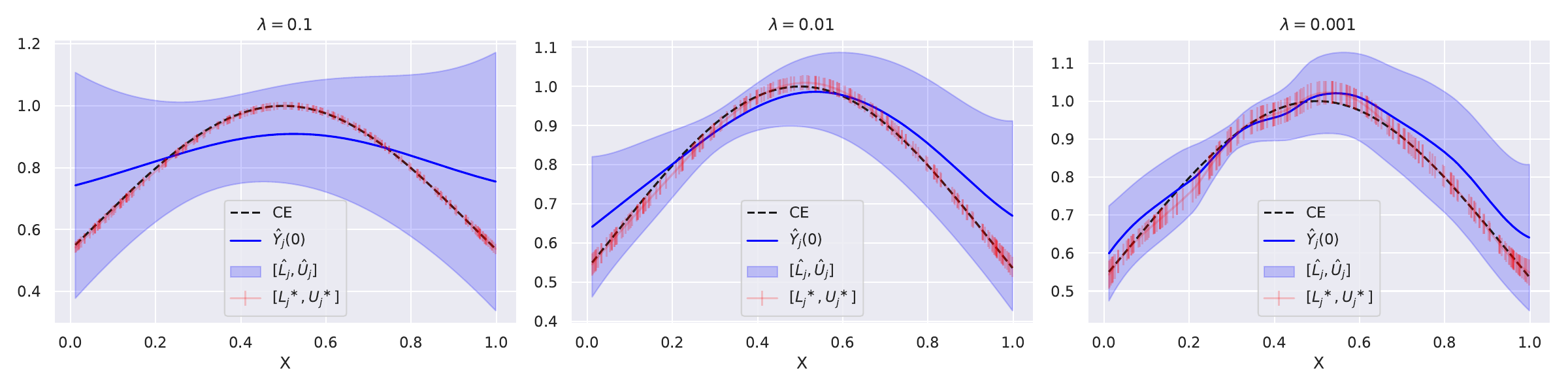}}
	\qquad
    \subfloat[$\sigma_0 = 1$]{\includegraphics[width=0.99\textwidth]{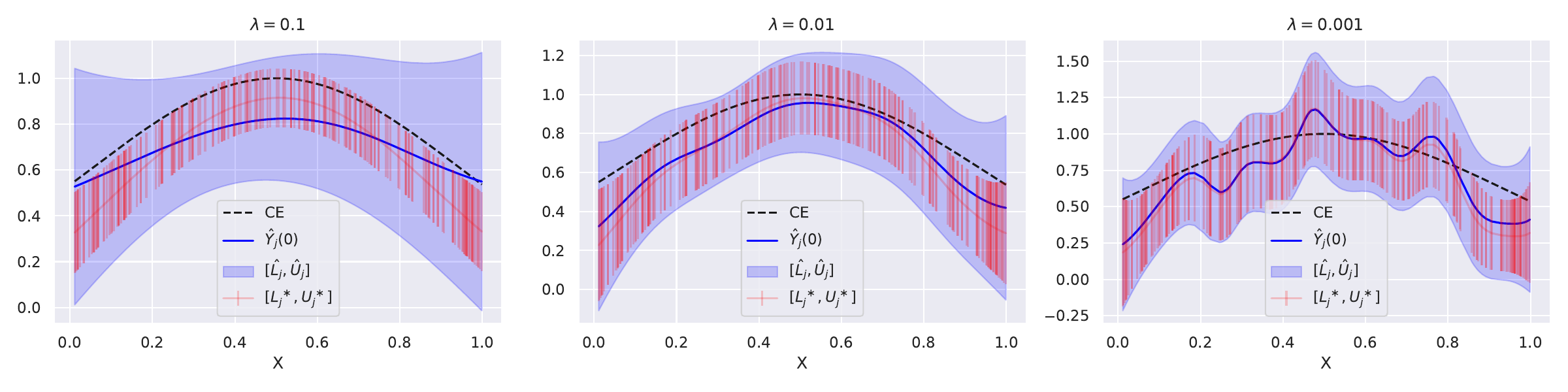}}
	\qquad
	\subfloat[$\sigma_0 = 3$]{\includegraphics[width=0.99\textwidth]{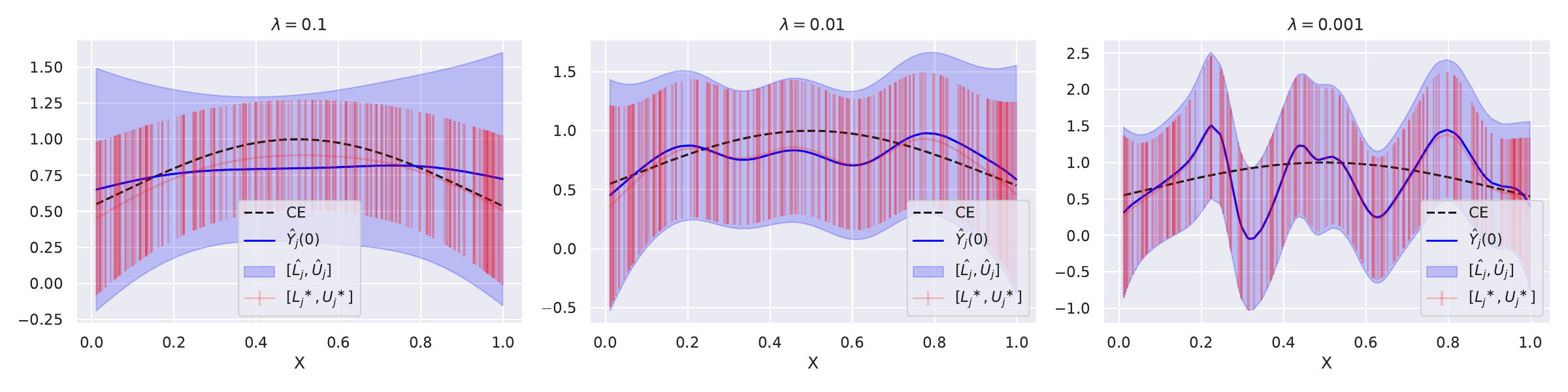}}
    \caption{\small Confidence intervals. The potential outcomes are generated by $Y_i(0) = f_0(x_i) + \epsilon_{0, i}$, where $\epsilon_{0, 1}, \ldots, \epsilon_{0, N}$ are independently drawn from $N(0, \sigma_0^2)$. The black dashed curve is the ground truth conditional expectation function $f_0(x) = \exp(-2.5 \times |x - 0.5|^2)$. The blue solid curve shows the estimated counterfactual outcome $\widehat{Y}_j(0) = \sum_{i \in \cC} \widehat{p}_{i j} Y_i$ for each treated unit $j$. The blue shaded region represents $[\widehat{L}_j, \widehat{U}_j]$ defined in \eqref{eq:lower} and \eqref{eq:upper}, where $\widehat{\theta}$ and $\widehat{\sigma}_0$ are estimated by the kernel ridge regression as explained in Section \ref{sec:confidence_intervals_method}. The oracle interval $[L_j^\ast, U_j^\ast]$ defined in \eqref{eq:lower_oracle} and \eqref{eq:upper_oracle} is shown using the red error bars.}
    \label{fig:CI_plots}
\end{figure}

Figure \ref{fig:CI_plots} shows the constructed confidence intervals $[\widehat{L}_j, \widehat{U}_j]$ for each treated unit $j$ for different values of $\sigma_0$. For each $\sigma_0 \in \{0.1, 1, 3\}$, we solve \eqref{eq:kernelized_formulation} for $\lambda \in \{0.1, 0.01, 0.001\}$ with uniform weights, run the kernel ridge regression to obtain $\widehat{\theta}, \widehat{\sigma}_0$, and produce $[\widehat{L}_j, \widehat{U}_j]$ which is shown as blue shaded regions in the plots. We compare this interval with an ``oracle'' interval $[L_j^\ast, U_j^\ast]$, which---as we will explain in the next section--- is guaranteed to have an exact $1 - \alpha$ coverage, namely,
\begin{equation}
	\label{eq:oracle_coverage}
	\P\left(f_0(x_j) \in [L_j^\ast, U_j^\ast] ~ | ~ \bX, \bZ\right) = 1 - \alpha.
\end{equation}
By construction, if the estimates $\widehat{\theta}, \widehat{\sigma}_0$ are accurate enough, the interval $[\widehat{L}_j, \widehat{U}_j]$ is supposed to contain the oracle interval $[L_j^\ast, U_j^\ast]$, which can be verified in Figure \ref{fig:CI_plots}. For each $\sigma_0$, decreasing $\lambda$ results in smaller approximation errors, namely, the bias correction reflected in $\left(K_{t t} + \widehat{P}^\top K_{c c} \widehat{P} - 2 K_{c t}^\top \widehat{P} \right)_{j j}^{1/2}$, and the estimated counterfactual outcomes (solid blue curves) that are more wiggly as the obtained weights have more variations. As we will see, $[\widehat{L}_j, \widehat{U}_j]$ is wider than $[L_j^\ast, U_j^\ast]$ by the twice of the bias correction, namely, 
\begin{equation*}
	(\widehat{U}_j - \widehat{L}_j) - (U_j^\ast - L_j^\ast) = 2 \cdot \widehat{\theta} \cdot \left(K_{t t} + \widehat{P}^\top K_{c c} \widehat{P} - 2 K_{c t}^\top \widehat{P} \right)_{j j}^{1/2}.
\end{equation*}
We can verify from Figure \ref{fig:CI_plots} that this gap between the two intervals gets smaller as $\lambda$ decreases, consistent with the above observation that the bias decreases.

Figure \ref{fig:CI_coverage} plots the coverage of $[\widehat{L}_j, \widehat{U}_j]$ and $[L_j^\ast, U_j^\ast]$, namely, the probability that they contain $f_0(x_j)$ estimated using the 1,000 samples. When $\lambda = 0.1$, we can see that the coverage of $[\widehat{L}_j, \widehat{U}_j]$ is almost $1$, meaning that the interval is too conservative, which results from the fact that the bias term is not small enough compared to the variance of the residuals. As $\lambda$ decreases, we can see that the coverage of $[\widehat{L}_j, \widehat{U}_j]$ gets closer to $1 - \alpha = 0.95$, which is consistent with the fact that the bias decreases. Particularly, when $\sigma_0 = 3$ and $\lambda = 0.001$, the coverage of $[\widehat{L}_j, \widehat{U}_j]$ is closest to the coverage of the oracle interval $[L_j^\ast, U_j^\ast]$, suggesting that the bias is dominated by the variance. Meanwhile, for $(\sigma_0, \lambda) = (1, 0.01)$ and $(\sigma_0, \lambda) = (3, 0.1)$, the coverage of $[\widehat{L}_j, \widehat{U}_j]$ is larger than the coverage of the oracle interval $[L_j^\ast, U_j^\ast]$ by approximately a constant, which suggests that the bias term and the variance term are nearly balanced. In the next section, we provide theoretical guidance to choose the regularization parameter $\lambda$ to balance the bias and variance tradeoff.

\begin{figure}[!ht]
    \centering
    \subfloat[$\sigma_0 = 0.1$]{\includegraphics[width=0.99\textwidth]{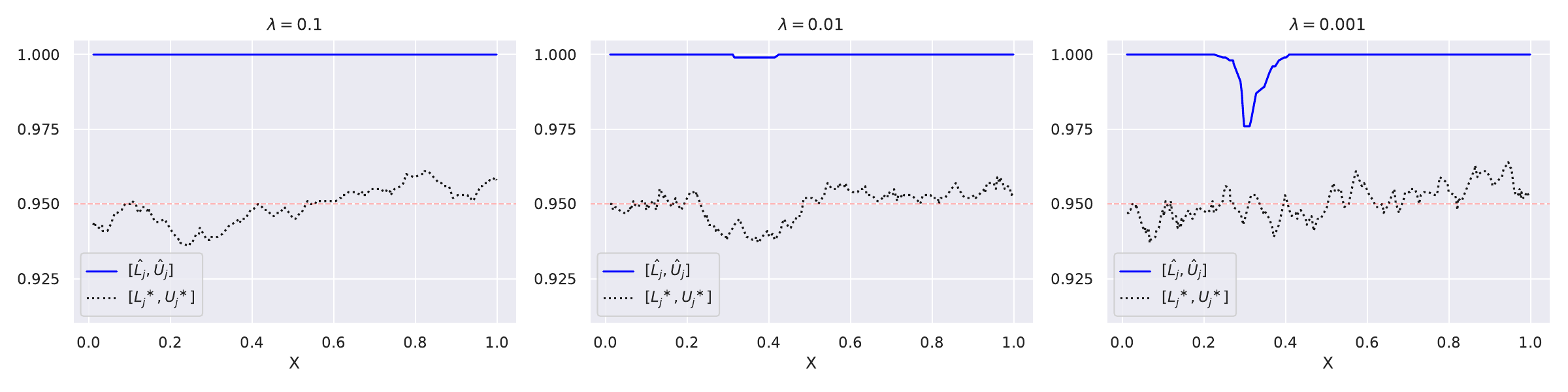}}
	\qquad
    \subfloat[$\sigma_0 = 1$]{\includegraphics[width=0.99\textwidth]{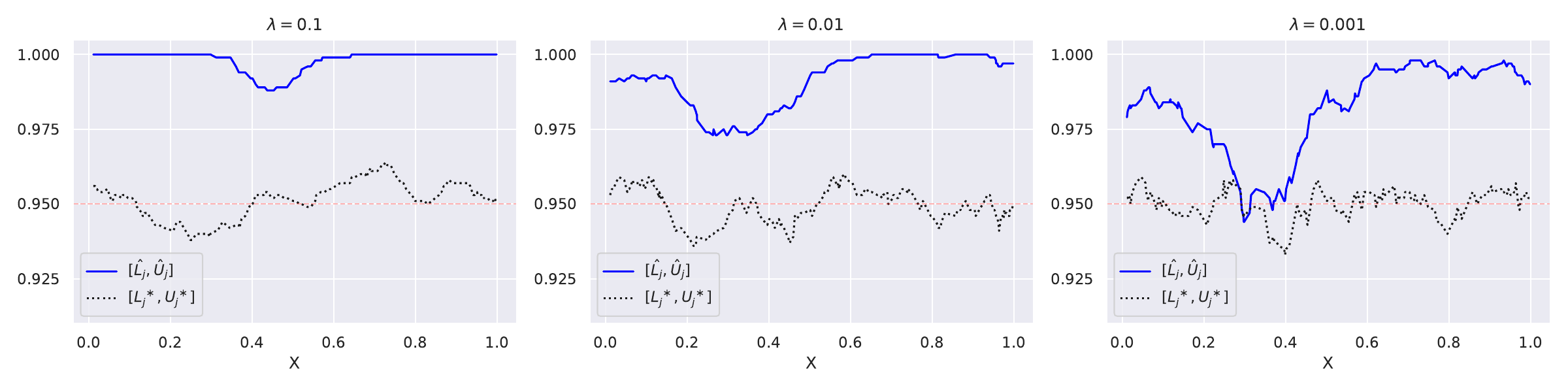}}
	\qquad
	\subfloat[$\sigma_0 = 3$]{\includegraphics[width=0.99\textwidth]{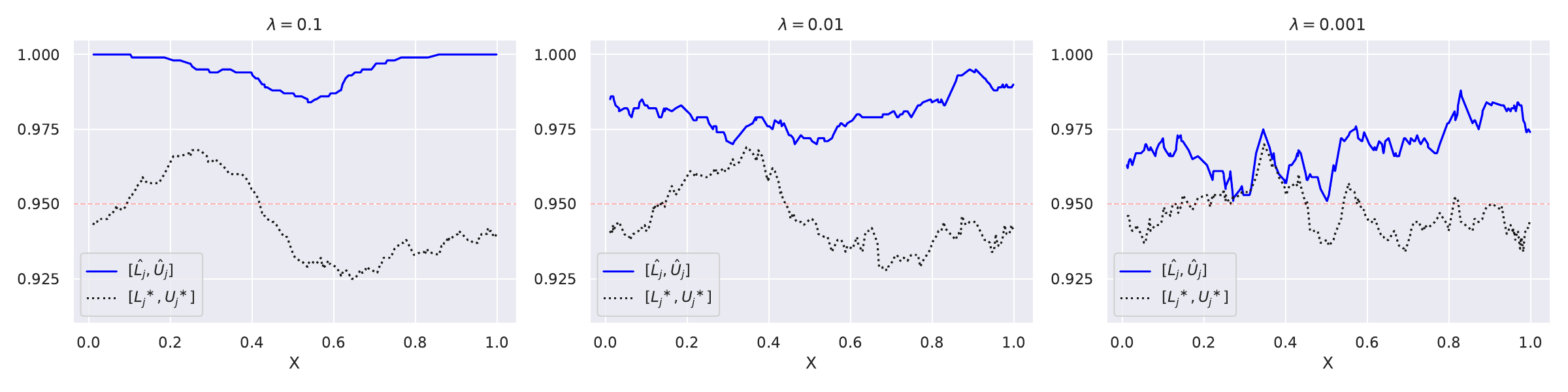}}
    \caption{\small Coverage of the confidence intervals $[\widehat{L}_j, \widehat{U}_j]$ and $[L_j^\ast, U_j^\ast]$ estimated using the 1,000 samples. The red dashed line shows the nominal coverage $1 - \alpha = 0.95$.}
    \label{fig:CI_coverage}
\end{figure}

\subsection{Coverage and Efficiency}
Now, we dive into each component in constructing individual confidence intervals to dissect the coverage and efficiency. We discuss how the entropic regularization parameter $\lambda$ drives the bias and variance tradeoff, as seen in the numerical simulations. Furthermore, we will show that optimizing the entropic-regularized convexified matching program \eqref{eq:kernelized_formulation} with the desired $\lambda$ chosen above is closely connected to optimizing the average of the squared lengths of the individual confidence intervals $\frac{1}{N_t} \sum_{j \in \cT} \mathrm{Len}_j^2$.

\paragraph{Point Estimate}
For each treated unit $j \in \cT$, the center of the confidence interval $[\widehat{L}_j, \widehat{U}_j]$ is the imputed counterfactual outcome $\widehat{Y}_j(0)$:
\begin{align}
	\widehat{Y}_j(0) = \sum_{i \in \cC} \widehat{p}_{i j} Y_i(0).
\end{align}
This weighted average estimator leverages all the information in $\bX$ to anchor an optimal coupling (or matching) between the treated and control units for synthesizing the controls. It uses convex combinations of the control units to synthesize the granular counterfactual outcomes for the treated units, while balancing the aggregate level statistics, the ATT. We remark that the optimal coupling $\widehat{\pi}$---and thus $\widehat{P} = N_t \widehat{\pi}$ ---is calculated based solely on $\bX$ and $\bZ$, not the outcomes $Y_i$'s. Recall the following bias-variance decomposition of the point estimate $\widehat{Y}_j(0)$: 
\begin{equation*}
	\widehat{Y}_j(0) - \E[Y_j(0) \, | \, \bX, \bZ] =  \underbrace{\widehat{Y}_j(0) - \E[\widehat{Y}_j(0) \, | \, \bX, \bZ]}_{:=\cV_j} + \underbrace{\E[\widehat{Y}_j(0) \, | \, \bX, \bZ] - \E[Y_j(0) \, | \, \bX, \bZ]}_{:=\cB_j}.
\end{equation*}

\paragraph{Bias-Awareness}
The point estimate could be biased as
\begin{equation*}
	\cB_j = \E[\widehat{Y}_j(0) \, | \, \bX, \bZ] - \E[Y_j(0) \, | \, \bX, \bZ] = \sum_{i \in \cC} \widehat{p}_{i j} f_0(x_i) - f_0(x_j) = \langle f_0, \sum_{i \in \cC} \widehat{p}_{i j} \phi_{x_i} - \phi_{x_j} \rangle_{\cH},
\end{equation*}
where the last equality uses the reproducing property of the RKHS. Since $f_0$ is unknown and estimating such a function uniformly on the support of the covariate vector can be difficult, we rely on a bias-awareness correction by invoking the Cauchy-Schwarz inequality:
\begin{equation}
	\label{eq:approximation_error_bound}
	|\cB_j| \le \|f_0\|_\cH \cdot \|\phi_{x_j} - \sum_{i \in \cC} \widehat{p}_{i j} \phi_{x_i}\|_\cH = \| f_0\|_{\cH} \cdot \left(K_{t t} + \widehat{P}^\top K_{c c} \widehat{P} - 2   K_{c t}^\top \widehat{P} \right)_{j j}^{1/2}.
\end{equation}
This explains subtracting and adding the term $\widehat{\theta} \cdot \left(K_{t t} + \widehat{P}^\top K_{c c} \widehat{P} - 2 K_{c t}^\top \widehat{P} \right)_{j j}^{1/2}$ in \eqref{eq:lower} and \eqref{eq:upper}, respectively, where $\widehat{\theta}$ estimates the unknown quantity $\|f_0\|_\cH$.

\paragraph{Variance}
Under the model,
\begin{equation}
	\label{eq:variance}
	\cV_j = \widehat{Y}_j(0) - \sum_{i \in \cC} \widehat{p}_{i j} f_0(x_i) = \sum_{i \in \cC} \widehat{p}_{i j} \epsilon_{0, i} \sim N(0, \sigma_0^2 \sum_{i \in \cC} \widehat{p}_{i j}^2).
\end{equation}
If we knew the bias $\cB_j$, we could construct the following oracle confidence interval $[L_j^\ast, U_j^\ast]$:
\begin{align}
	L_j^\ast & := \widehat{Y}_j(0) - \cB_j - z_{1 - \frac{\alpha}{2}} \cdot \sigma_0 \sqrt{\sum_{i \in \cC} \widehat{p}_{i j}^2}, \label{eq:lower_oracle} \\
	U_j^\ast & := \widehat{Y}_j(0) - \cB_j + z_{1 - \frac{\alpha}{2}} \cdot \sigma_0 \sqrt{\sum_{i \in \cC} \widehat{p}_{i j}^2}, \label{eq:upper_oracle}
\end{align}
Due to \eqref{eq:variance}, the oracle interval $[L_j^\ast, U_j^\ast]$ is guaranteed to have the exact $1 - \alpha$ coverage, namely, \eqref{eq:oracle_coverage} holds. However, since $\cB_j$ is unknown, we rely on the aforementioned bias-awareness correction to construct the confidence interval $[\widehat{L}_j, \widehat{U}_j]$, obtained by inflating the oracle interval using \eqref{eq:approximation_error_bound}.

\paragraph{Efficiency: Interval Length}
We now explain how the width of the confidence interval is related to our convexified matching objective. As a result, the explanation will shed light on why the entropic regularization serves as a tuning parameter to trade off bias and variance for overall interval length efficiency. First, we need the following simple facts on the Kullback-Leibler divergence, Hellinger distance, and $\chi^2$ distance.
\begin{proposition}
	\label{prop:Hellinger-KL-Chisq}
	For any probability vector $p \in \Delta_n$, 
	\begin{equation*}
		\frac{1}{2(n \|p\|_\infty + 1)} \le \frac{\sum_{i = 1}^{n} p_i \log(p_i) + \log(n)}{n \sum_{i = 1}^{n} p_i^2 - 1} \le 1.
	\end{equation*}
\end{proposition}
The above fact will help us relate the entropy term and the standard Euclidean norm of the weights. For a fixed $j \in \cT$, let us consider that $(\widehat{p}_{i j})_{i \in \cC} \in \Delta_{N_c}$, the $j$-th column of $\widehat{P}$, is delocalized in the sense $\max_{i \in \cC} \widehat{p}_{i j} \le (M/2-1)/ N_c$ with some constant $M > 4$. The interval length is as follows:
\begin{align*}
	\mathrm{Len}_j := 2 \cdot \| f_0 \|_{\cH} \cdot \|\phi_{x_j} - \sum_{i \in \cC} \widehat{p}_{i j} \phi_{x_i}\|_\cH + 2 \cdot z_{1 - \frac{\alpha}{2}} \cdot \sigma_0 \sqrt{\sum_{i \in \cC} \widehat{p}_{i j}^2}.
\end{align*}
From $\widehat{\pi}_{i j} = \widehat{p}_{i j} / N_t$ and Proposition \ref{prop:Hellinger-KL-Chisq}, we deduce the following bounds:
\begin{align*}
	\mathrm{Len}_j^2 & \le 8 \left\{\|f_0\|_{\cH}^2  \|\phi_{x_j} - \sum_{i \in \cC} \tfrac{\pi_{i j}}{1 /N _t} \phi_{x_i}\|_\cH^2 + M \cdot z^2_{1 - \frac{\alpha}{2}} \sigma_0^2 \frac{ N_t}{N_c} \left[\sum_{i \in \cC} \widehat{\pi}_{i j} \log \tfrac{\widehat{\pi}_{i j}}{e} + \frac{\tfrac{1}{M} + 1 + \log(N_c N_t)}{N_t}\right]\right\} ,\\
	\mathrm{Len}_j^2 & \ge 4 \left\{\|f_0\|_{\cH}^2 \|\phi_{x_j} - \sum_{i \in \cC} \tfrac{\pi_{i j}}{1 / N_t} \phi_{x_i}\|_\cH^2 + z^2_{1 - \frac{\alpha}{2}} \sigma_0^2 \frac{N_t}{N_c} \left[\sum_{i \in \cC} \widehat{\pi}_{i j} \log \tfrac{\widehat{\pi}_{i j}}{e} + \frac{2 + \log(N_c N_t)}{N_t}\right]\right\} .
\end{align*}
Averaging over $j \in \cT$, and barring the constant $M$, the upper and lower bounds on the right-hand sides remind us of the entropic regularized convexified matching objective in \eqref{eq:kernelized_formulation}, with the regularization parameter 
\begin{equation*}
	\lambda \asymp \frac{z^2_{1 - \frac{\alpha}{2}}}{N_c} \frac{\sigma_0^2}{\|f_0\|_{\cH}^2}.
\end{equation*}
The above tradeoff confirms the empirical findings: for a high signal-to-noise problem, to balance bias and variance, we need a smaller $\lambda$ for better approximation error; for a low signal-to-noise problem, a larger $\lambda$ is preferred to reduce variance. Such a phenomenon confirms the observations in Figure~\ref{fig:CI_plots}. Searching for a coupling minimizing the objective of convexified matching is closely related to minimizing the interval length, provided the regularization parameter $\lambda$ is appropriately chosen. Therefore, optimizing the entropic regularized convexified matching program \eqref{eq:kernelized_formulation} with the desired $\lambda$ chosen above is closely connected to optimizing the average of the squared lengths of the individual confidence intervals $\frac{1}{N_t} \sum_{j \in \cT} \mathrm{Len}_j^2$.

\section{Optimization: Algorithms and Analysis}
\label{sec-optimization}
This section introduces and analyzes optimization algorithms to solve the convexified matching problem, KSC defined in \eqref{eq:kernelized_formulation} with uniform weights $v = \frac{1}{N_t}$ on the treated, as shown in \eqref{eq:kernelized_formulation_uniform_treated}. Notice that it is an instance of the constrained nonlinear optimization problem: for $n, m \in \N$, $a \in \Delta_n^+$, $b \in \Delta_m^+$, $\lambda > 0$, and $g \colon \R^{n \times m} \to \R$,
\begin{equation}
	\label{eq:optim}
	\min_{\pi \in \Pi_{a, b}} ~ g(\pi) + \lambda h(\pi),
\end{equation}
where $\Pi_{a, b} = \{\pi \in \R_+^{n \times m} : \pi 1_m = a ~ \text{and} ~ \pi^\top 1_n = b\}$ and $h \colon \R_+^{n \times m} \to \R$ is the entropy function defined by $h(\pi) = \sum_{i = 1}^{n} \sum_{j = 1}^{m} \pi_{i j} \log \frac{\pi_{i j}}{e}$. The convexified matching amounts to the case where $g$ is a convex quadratic function, in which \eqref{eq:optim} is a convex program taking the form \eqref{eq:kernelized_formulation_uniform_treated}. It can be solved using the interior point method in $\mathrm{poly}(nm)$ time, possibly slow in practice for coupling matrices with $nm$ large. 

In this section, we devise a faster algorithm that requires an almost \textit{dimension-free} number of oracle calls to a subroutine---a simple iterative matrix scaling procedure called the Sinkhorn algorithm \cite{sinkhorn1967diagonal}. When $\lambda$ is sufficiently large, the algorithm compares favorably to the interior point method, requiring only a $\log(1/\epsilon)$ number of Sinkhorn subroutines. Later, we modify and extend the algorithm to cover all $\lambda \in \R_+$ by connecting to the steepest descent method under the Kullback-Leibler divergence. The convergence is admittedly slower for small $\lambda$: $\log(n m)/\epsilon$ number of Sinkhorn subroutines is required, but it still compares favorably to the interior point method.

\paragraph{Additional notation} For a matrix $A \in \R^{n \times m}$, let $\|A\|_{\mathrm{F}}$ denote its Frobenius norm, $\|A\|_1 := \sum_{i = 1}^{n} \sum_{j = 1}^{m} |A_{i j}|$ denote its $L^1$ norm, and $\|A\|_\infty := \max_{1 \le i \le n} \max_{1 \le j \le m} |A_{i j}|$ denote its $L^\infty$ norm. Let $h$ keep its definition as the entropy function, and let $\nabla h(A)$ denote the gradient of $h$, which is $\log(A)$, the entrywise logarithm of $A$, provided all entries of $A$ are positive. For any $a \in \Delta_n$ and $b \in \Delta_m$, let $\Pi_{a, b}^{+} := \{A \in \Pi_{a, b} : A_{i j} > 0 ~~ \forall i, j\}$. Let $\Delta_{n, m} = \{A \in \R_+^{n \times m} : \sum_{i = 1}^{n} \sum_{j = 1}^{m} A_{i j} = 1\}$. For vectors $u, v \in \R^n$, let $\frac{u}{v}$ denote the entrywise division provided all the entries of $v$ are nonzero. 

\subsection{Optimality Conditions and the Sinkhorn Algorithm}
When $g$ is convex and $\lambda > 0$, the strong convexity of $h$ on $\Pi_{a, b}$ implies that \eqref{eq:optim} admits a unique minimizer. Moreover, as $h$ prevents the minimizer from having a zero entry, \eqref{eq:optim} admits a unique minimizer on $\Pi_{a, b}^{+}$, which is the interior of $\Pi_{a, b}$. It turns out that this unique minimizer of \eqref{eq:optim} is the fixed point of an operator related to the entropic regularized optimal transport problem \cite{cuturi2013sinkhorn}, as we shall show in Proposition~\ref{prop:minimizer=fixed_point}.  Later, in Theorem~\ref{thm:fixed-point_L1}, we derive convergence to the minimizer via the iterative matrix scaling subroutine. This subroutine is referred to as the Sinkhorn algorithm \cite{sinkhorn1967diagonal}, which we introduce below.

\begin{definition}[Sinkhorn \cite{sinkhorn1967diagonal,cuturi2013sinkhorn}]
	\label{prop:EOT_Sinkhorn}
	Fix $\lambda > 0$. The following operator $\Phi_\lambda \colon \R^{n \times m} \to \Pi_{a, b}^{+}$ is well-defined:
	\begin{equation*}
		\Phi_\lambda(C) := \argmin_{\pi \in \Pi_{a, b}} \left(\langle C, \pi \rangle + \lambda h(\pi)\right) \quad \forall C \in \R^{n \times m}.
	\end{equation*}
	The right-hand side is often called the entropic regularized optimal transport problem, given a cost matrix $C$. For any $C \in \R^{n \times m}$, suppose there are $\mu \in \R^n$ and $\nu \in \R^m$ such that 
	\begin{equation*}
		\exp\left(-\frac{\mu 1_m^\top + 1_n \nu^\top + C}{\lambda}\right) \in \Pi_{a, b},
	\end{equation*}
	where $\exp$ is applied entrywise. Then, the following must hold:
	\begin{equation*}
		\exp\left(-\frac{\mu 1_m^\top + 1_n \nu^\top + C}{\lambda}\right) = \Phi_\lambda(C).
	\end{equation*}
	Moreover, for any $C \in \R^{n \times m}$, one can obtain $\Phi_\lambda(C)$ by iteratively scaling the rows and columns of the matrix $\exp(-C / \lambda)$ using the Sinkhorn algorithm summarized in Algorithm \ref{alg:Sinkhorn}, namely, $\Phi_\lambda(C)$ is the limit of the sequence produced from $\mathrm{Sinkhorn}(e^{-C / \lambda}, a, b)$. 
\end{definition}

\begin{algorithm}[ht]
    \caption{$\mathrm{Sinkhorn}(P, a, b)$}
    \label{alg:Sinkhorn}
    \begin{algorithmic}[1]
    \REQUIRE A matrix $P \in \R^{n \times m}$ with positive entries, $a \in \Delta_n^+$, and $b \in \Delta_m^+$.
    \STATE Initialize $k \leftarrow 0$, $x^{(k)} \leftarrow 1_n$, and $y^{(k)} \leftarrow \frac{b}{P^\top x^{(k)}}$.
	\REPEAT
		\STATE 
		\begin{equation*}
			x^{(k + 1)} \leftarrow \frac{a}{P y^{(k)}} \quad \text{and} \quad y^{(k + 1)} \leftarrow \frac{b}{P^\top x^{(k + 1)}}.
		\end{equation*}
		\STATE $P^{(k + 1)} = \mathrm{diag}(x^{(k + 1)}) P \mathrm{diag}(y^{(k + 1)})$.
		\STATE $k \leftarrow k + 1$ 
	\UNTIL{Discrepancy between $P^{(k)} 1_m$ and $a$ reaches a desired level.}
    \RETURN $P^{(k)}$.
\end{algorithmic}
\end{algorithm}

For simplicity, Algorithm \ref{alg:Sinkhorn} states the Sinkhorn algorithm such that the column sum of $P^{(k)}$ matches with $b$, namely, $(P^{(k)})^\top 1_n = b$ for any $k \ge 0$. Then, to assess the accuracy of the scaled matrix $P^{(k)}$, we may only need to check a suitable discrepancy between $a$ and the row sum of $P^{(k)}$, that is, $P^{(k)} 1_m$.

The following proposition shows that when $g$ is convex, the unique minimizer of \eqref{eq:optim} is the fixed point of an operator given by the composition of $\Phi_\lambda$ and $\nabla g$.

\begin{proposition}
	\label{prop:minimizer=fixed_point}
	Fix $\lambda > 0$. For a function $g \colon \R^{n \times m} \to \R$ such that $\nabla g$ exists, define an operator $T_\lambda \colon \R^{n \times m} \to \Pi_{a, b}^{+}$ defined by $T_\lambda = \Phi_\lambda \circ \nabla g$, that is,
	\begin{equation}
		\label{eq:operator_T}
		T_\lambda(A) = \argmin_{\pi \in \Pi_{a, b}} (\langle \nabla g(A), \pi \rangle + \lambda h(\pi)) \quad \forall A \in \R^{n \times m}.
	\end{equation}
	If $g$ is convex, \eqref{eq:optim} admits a unique minimizer contained in $\Pi_{a, b}^{+}$, say, $\pi^\star \in \Pi_{a, b}^{+}$, and $\pi^\star$ is the unique fixed point of $T_\lambda$, namely, $\pi^\star = T_\lambda(\pi^\star)$. 
\end{proposition}

Based on the established connection between the optimality condition of \eqref{eq:optim} and the Sinkhorn algorithm, we propose two algorithms to solve \eqref{eq:optim} in the following sections. 

\subsection{Fixed-Point Algorithm}
By Proposition \ref{prop:minimizer=fixed_point}, finding the fixed point of $T_\lambda$ is equivalent to solving \eqref{eq:optim}. Therefore, we propose solving \eqref{eq:optim} by the fixed-point iterations based on the operator $T_\lambda \colon \R^{n \times m} \to \Pi_{a, b}^{+}$, which is simply choosing the initial point $\pi^{(0)} \in \Pi_{a, b}$ properly and iterating $\pi^{(k + 1)} \leftarrow T_\lambda(\pi^{(k)})$ for $k \ge 0$, where $T_\lambda(\pi^{(k)}) = \Phi_\lambda(\nabla g(\pi^{(k)}))$ is approximated by $\mathrm{Sinkhorn}(e^{-\nabla g(\pi^{(k)}) / \lambda}, a, b)$ as explained in Definition \ref{prop:EOT_Sinkhorn}. Algorithm \ref{alg:main} summarizes this procedure.

\begin{algorithm}[ht]
    \caption{$\mathrm{FixedPoint}(g, \lambda, a, b)$}
    \label{alg:main}
    \begin{algorithmic}[1]
        \REQUIRE A map $g \colon \R^{n \times m} \to \R$, a parameter $\lambda > 0$, $a \in \Delta_n^+$, and $b \in \Delta_m^+$.
        \STATE Pick any $\pi^{(0)} \in \Pi_{a, b}$ and set $k \leftarrow 0$.
        \REPEAT 
            \STATE $\pi^{(k + 1)} \leftarrow \mathrm{Sinkhorn}(e^{-\nabla g(\pi^{(k)}) / \lambda}, a, b)$.
            \STATE $k \leftarrow k + 1$. 
        \UNTIL{Discrepancy between $\pi^{(k)}$ and $\pi^{(k - 1)}$ reaches a desired level.}
        \RETURN $\pi^{(k)}$.
    \end{algorithmic}
\end{algorithm}

We show that $T_\lambda$ is a contraction if the gradient of $g$ is Lipschitz and $\lambda$ is larger than the Lipschitz constant. In such a case, Algorithm \ref{alg:main} converges to the fixed point quickly. When $g$ is a quadratic function, we can write the Lipschitz constant in terms of the $L^\infty$ norm of the Hessian matrix, which is independent of the input dimension $n m$. We do not require $g$ to be convex for this result; $T_\lambda$ admits a unique fixed point due to the contraction argument, which is independent of Proposition \ref{prop:minimizer=fixed_point}. The role of Proposition \ref{prop:minimizer=fixed_point} is to translate this convergence result into the convergence to the minimizer of \eqref{eq:optim} for convex $g$ by equating the fixed point to the minimizer.

\begin{theorem}
	\label{thm:fixed-point_L1}
	Let $g \colon \R^{n \times m} \to \R$ be a quadratic function, with some $H \in \R^{n \times n}$ that is symmetric and $C \in \R^{n \times m}$,
	\begin{equation*}
		g(\pi) = \frac{1}{2} \langle \pi, H \pi \rangle + \langle C, \pi \rangle \quad \forall \pi \in \R^{n \times m}.
	\end{equation*}
	Assume $\lambda > \|H\|_\infty$. Then, the operator $T_\lambda$ defined in \eqref{eq:operator_T} is a contraction under the distance defined by $\|\cdot\|_1$ and has a unique fixed-point, say, $\pi^\star \in \Pi_{a, b}^{+}$. Moreover, Algorithm \ref{alg:main}, assuming the inner loop $\mathrm{Sinkhorn}$ is always exact, outputs a sequence $(\pi^{(k)})_{k \ge 0}$ such that for any $T \in \N$,
	\begin{equation}
		\label{eq:fixed_point_bound_L1}
		\|\pi^{(T)} - \pi^\star\|_1 \le \frac{(\|H\|_\infty / \lambda)^T}{1 - (\|H\|_\infty / \lambda)} \|\pi^{(1)} - \pi^{(0)}\|_1.
	\end{equation}
\end{theorem}

\begin{remark}
	Theorem \ref{thm:fixed-point_L1} shows that when $\lambda > \|H\|_\infty$, to obtain an $\epsilon$-approximate fixed point $\|\pi^{(T)} - \pi^\star\|_1^2 \leq \epsilon$ , we need $\log(1 / \epsilon) / 2\log(\lambda / \|H\|_\infty)$ number of Sinkhorn routines---a dimension-free quantity.
\end{remark}

\subsection{Local Algorithm with Kullback-Leibler Geometry}
The fixed-point algorithm may not converge for sufficiently small $\lambda$. To fill in the gap when $\lambda \in [0, \| H\|_{\infty})$ left unanswered by Theorem~\ref{thm:fixed-point_L1}, we employ the steepest descent method under Bregman divergence to solve \eqref{eq:optim} as follows: let $f = g + \lambda h$,
\begin{equation}
	\label{eqn-steepest-descent}
    \pi^{(k + 1)} = \argmin_{\pi \in \Pi_{a, b}} \left(\langle \nabla f(\pi^{(k)}), \pi \rangle + \frac{D_h(\pi, \pi^{(k)})}{\tau_k}\right), 
\end{equation}
where $D_h(\pi, \pi^{(k)}) = h(\pi) - h(\pi^{(k)}) - \langle \nabla h(\pi^{(k)}), \pi - \pi^{(k)} \rangle$ is the Bregman divergence, which equals the Kullback-Leibler divergence on $\Pi_{a, b}$, and $\tau_k > 0$ is a suitable step size. With the entropic regularization $h$, we implement the steepest descent over the polytope $\Pi_{a, b}$ under the Kullback-Leibler divergence. We emphasize that this can be solved using the Sinkhorn algorithm as well since \eqref{eqn-steepest-descent} is equivalent to
\begin{equation*}
	\begin{split}
		\pi^{(k + 1)} 
		& = \argmin_{\pi \in \Pi_{a, b}} \left(\langle \nabla g(\pi^{(k)}) + (\lambda - \tau_k^{-1}) \nabla h(\pi^{(k)}), \pi \rangle + \tau_k^{-1} h(\pi)\right) \\
		& = \Phi_{\tau_k^{-1}}\left(\nabla g(\pi^{(k)}) + (\lambda - \tau_k^{-1})  \nabla h(\pi^{(k)})\right).
	\end{split}
\end{equation*}
The above can be approximated by applying the Sinkhorn algorithm to $\nabla g(\pi^{(k)}) + (\lambda - \tau_k^{-1})  \nabla h(\pi^{(k)})$. Algorithm \ref{alg:sdkl} summarizes this procedure. Note that letting $\tau_k = \lambda^{-1}$ for all $k \ge 0$ in Algorithm \ref{alg:sdkl} recovers Algorithm \ref{alg:main}.

\begin{algorithm}
    \caption{$\mathrm{SteepestDescentKL}(g, \lambda, \{\tau_k\}_{k \ge 0}, a, b)$}
    \label{alg:sdkl}
    \begin{algorithmic}[1]
        \REQUIRE A map $g \colon \R^{n \times m} \to \R$, a parameter $\lambda > 0$, $a \in \Delta_n^+$, and $b \in \Delta_m^+$.
        \REQUIRE Step sizes $\{\tau_k\}_{k \ge 0}$.
		\STATE Pick any $\pi^{(0)} \in \Pi_{a, b}$ and set $k \leftarrow 0$.
        \REPEAT 
            \STATE $\pi^{(k + 1)} \leftarrow \mathrm{Sinkhorn}(e^{-(\nabla g(\pi^{(k)}) + (\lambda - \tau_k^{-1}) \nabla h(\pi^{(k)})) / \tau_k^{-1}}, a, b)$.
            \STATE $k \leftarrow k + 1$. 
        \UNTIL{Discrepancy between $\pi^{(k)}$ and $\pi^{(k - 1)}$ reaches a desired level.}
        \RETURN $\pi^{(k)}$.
    \end{algorithmic}
\end{algorithm}

We show that Algorithm \ref{alg:sdkl} outputs a sequence that converges to the minimum of \eqref{eq:optim} provided the step size is below a certain threshold.
\begin{theorem}
	\label{thm:steepest-descent-KL-convergence}
	Let $g \colon \R^{n \times m} \to \R$ be a convex quadratic function, for some $H \in \R^{n \times n}$ that is symmetric and positive semidefinite and $C \in \R^{n \times m}$,
	\begin{equation*}
		g(\pi) = \frac{1}{2} \langle \pi, H \pi \rangle + \langle C, \pi \rangle \quad \forall \pi \in \R^{n \times m}.
	\end{equation*}
	Then, for any $\lambda \geq 0$, if Algorithm \ref{alg:sdkl}, assuming the inner loop $\mathrm{Sinkhorn}$ is always exact, with constant stepsize $\tau_k = \tau$ for all $k \ge 0$, where $\tau^{-1} \geq \|H\|_\infty + \lambda$, it outputs a sequence $(\pi^{(k)})_{k \ge 0}$ such that for any $T \in \N$,
	\begin{equation}
		\label{eq:steepest-descent-KL-convergence}
		g(\pi^{(T)}) + \lambda h(\pi^{(T)}) - \min_{\pi \in \Pi_{a, b}} (g(\pi) + \lambda h(\pi)) \le \frac{1}{T} \frac{\log(n m)}{\tau}.
	\end{equation}
	For $\lambda >0$, the following holds as well:
	\begin{equation}
		\label{eq:convergence_to_minimizer}
		\|\pi^{(T)} - \pi^\star\|_1^2 \leq \frac{1}{T} \frac{2 \log(n m)}{\lambda \tau},
	\end{equation}
	where $\pi^\star \in \Pi_{a, b}^+$ is the unique minimizer of \eqref{eq:optim}.
\end{theorem}

\begin{remark}
	The $\lambda \in [\| H\|_{\infty}, \infty)$ case has been studied in Theorem~\ref{thm:fixed-point_L1}.
	Theorem~\ref{thm:steepest-descent-KL-convergence} shows that for any $\lambda \in [0, \| H\|_{\infty})$, with the choice $\tau^{-1} = 2\|H\|_\infty$, we need $1/\epsilon \cdot 2\| H\|_{\infty}\log(n m)$ number of Sinkhorn subroutines to get $\epsilon$-close to the objective value. Further more, when $\lambda$ is strictly positive, we can assure that after $1/\epsilon \cdot 4\| H\|_{\infty}\log(n m)/\lambda$ number of Sinkhorn subroutines, $\| \pi^{(T)} - \pi^\star \|_1^2 \leq \epsilon$, with a logarithmic dependence on the problem dimension.
\end{remark}

\section{Applications: Revisit the NSW Data}
\label{sec-application}
In this section, we numerically investigate our imputation and individualized inference method using the NSW data and compare it with other methods, continued from Section \ref{sec:numerical_illustration}.

\subsection{Imputation}
We conduct a robustness check of the imputation results in Section \ref{sec:numerical_illustration} on the NSW data by specifying different nonlinear kernels. We consider the RBF kernel and the polynomial kernel of degree two. The results are shown in Figure \ref{fig:NSW_imputation_scatterplots_rbf} and Figure \ref{fig:NSW_imputation_scatterplots_poly}. Though the overall patterns are similar to those using the linear kernel, we can observe some differences. First, for the experimental data, the imputed counterfactual outcomes via the RBF kernel with $\lambda = 0.001$ show the largest variability. In Figure \ref{fig:NSW_imputation_scatterplots_rbf}(a), there are two treated units whose imputed counterfactual outcomes are above 15000, which do not occur in the results via the linear or polynomial kernel; as a result, the right tail of the histogram along the $x$-axis is longer than the other cases. For the PSID data, we can see that using the polynomial kernel leads to the most spread out imputed outcomes, visualized by the right tails of the histograms along the $x$-axis in (a) and (b) of Figure \ref{fig:NSW_imputation_scatterplots_poly} which are longer than the other cases. We point out that unlike the NN-matching method and our KSC method, whose counterfactual distributions are supported on $\R_+$ even for the nonexperimental PSID data, the regression imputation method presents a wider support for imputed values, even with negative outcomes as seen across Figures~\ref{fig:NSW_imputation_scatterplots}, \ref{fig:NSW_imputation_scatterplots_rbf} and \ref{fig:NSW_imputation_scatterplots_poly}.

\begin{figure}[!ht]
	\centering
	\subfloat[NSW $\lambda = 0.001$]{\includegraphics*[width=0.47\textwidth]{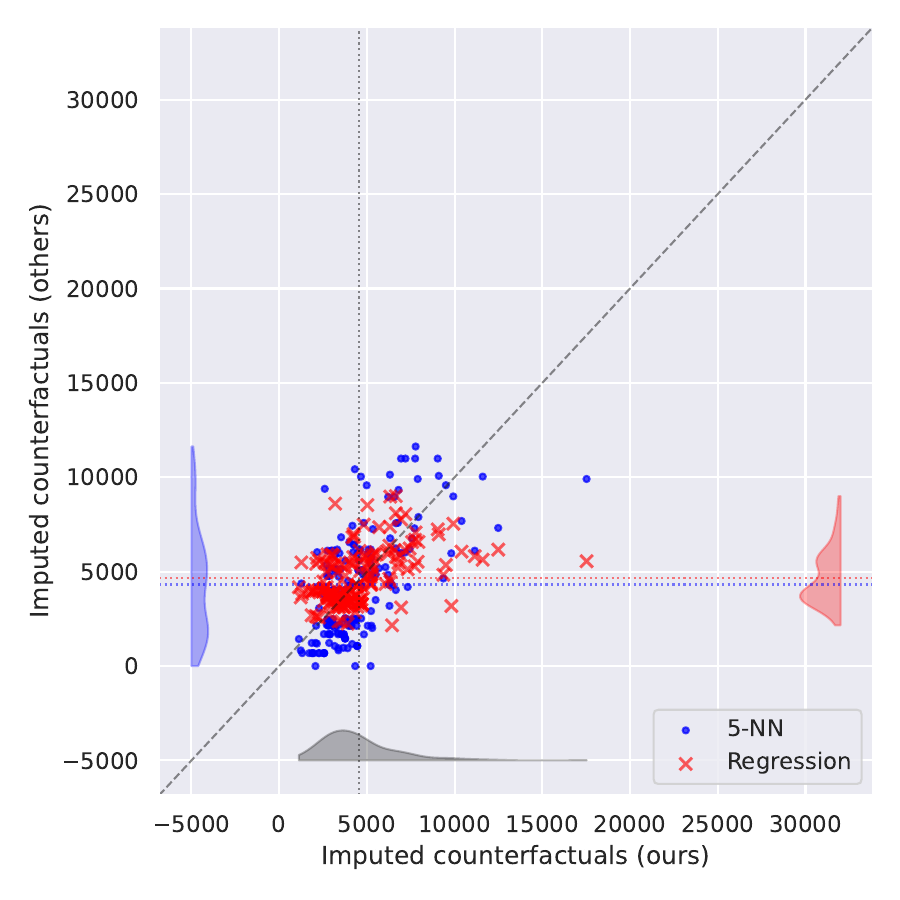}}
	\subfloat[NSW $\lambda = 0.01$]{\includegraphics*[width=0.47\textwidth]{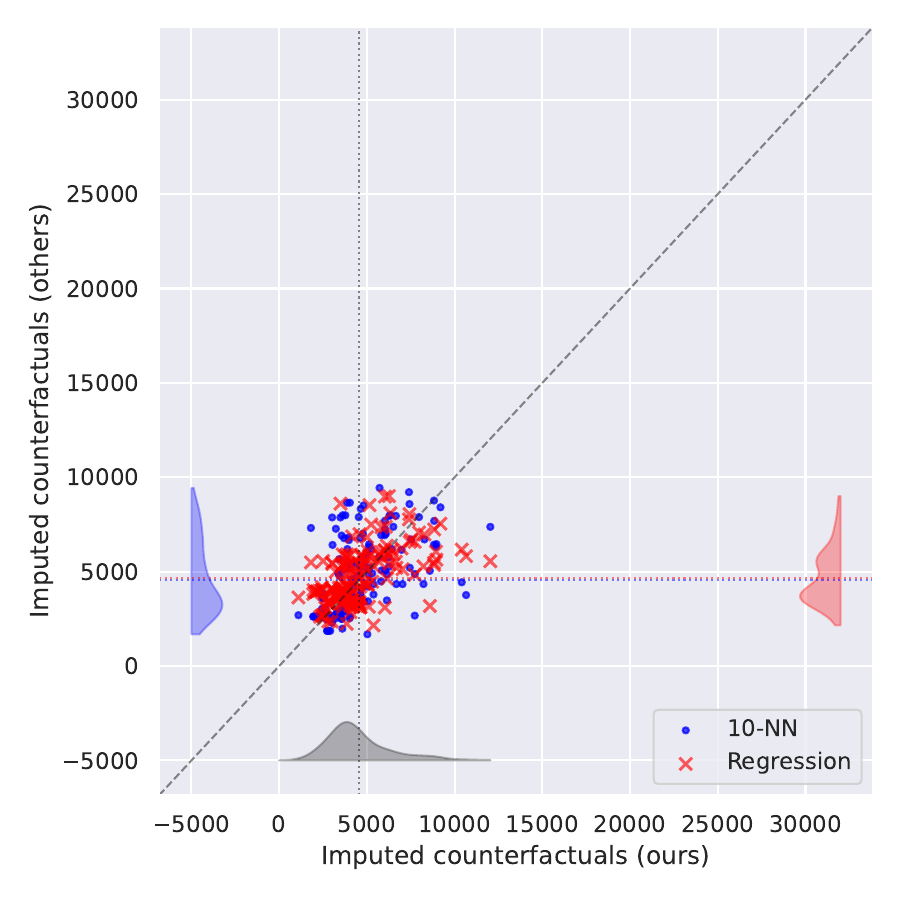}}
	\qquad
	\subfloat[PSID $\lambda = 0.001$]{\includegraphics*[width=0.47\textwidth]{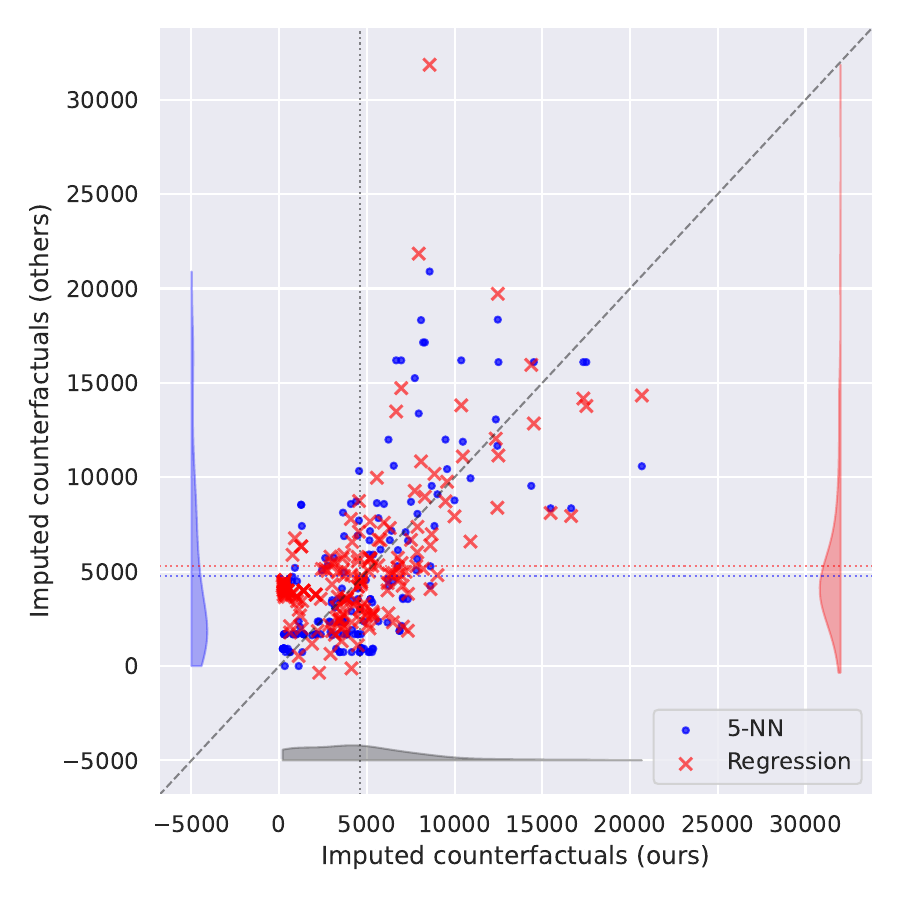}}
	\subfloat[PSID $\lambda = 0.01$]{\includegraphics*[width=0.47\textwidth]{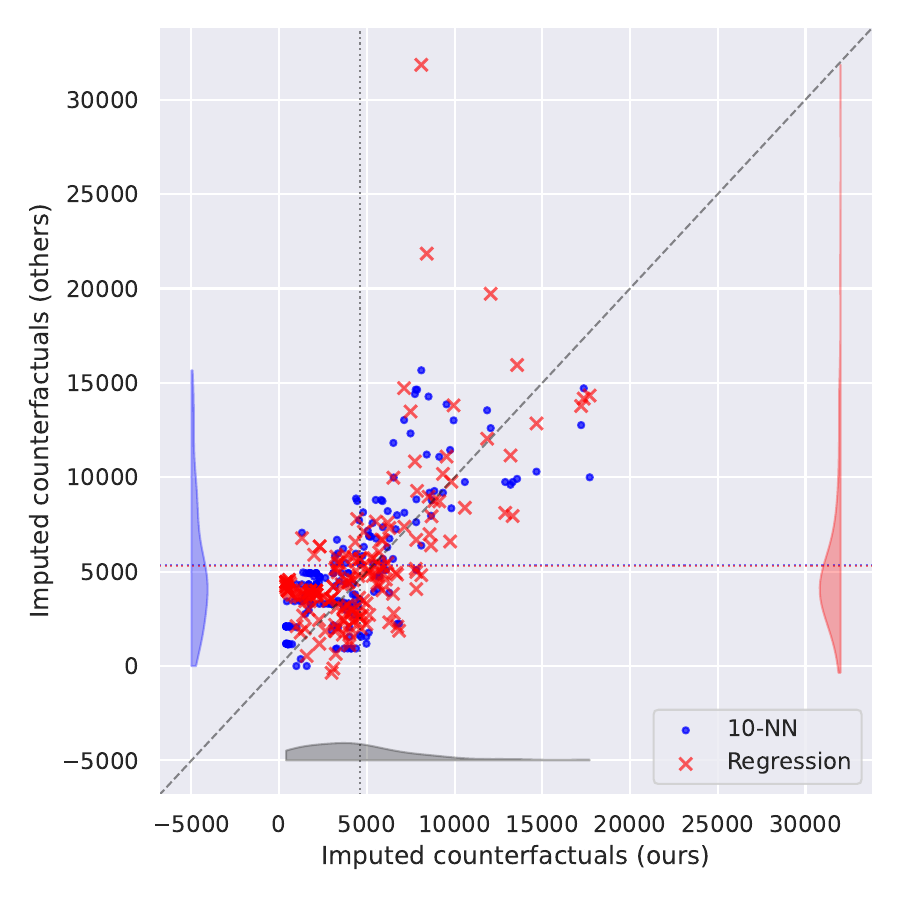}}
	\caption{\small Scatter plots of the imputed counterfactual outcomes of the treated along with the histograms of the marginal distributions. The setup is the same as in Figure \ref{fig:NSW_imputation_scatterplots}, but the results are based on the RBF kernel for both NSW and PSID data.}
	\label{fig:NSW_imputation_scatterplots_rbf}
\end{figure}

\begin{figure}[!ht]
	\centering
	\subfloat[NSW $\lambda = 0.001$]{\includegraphics*[width=0.47\textwidth]{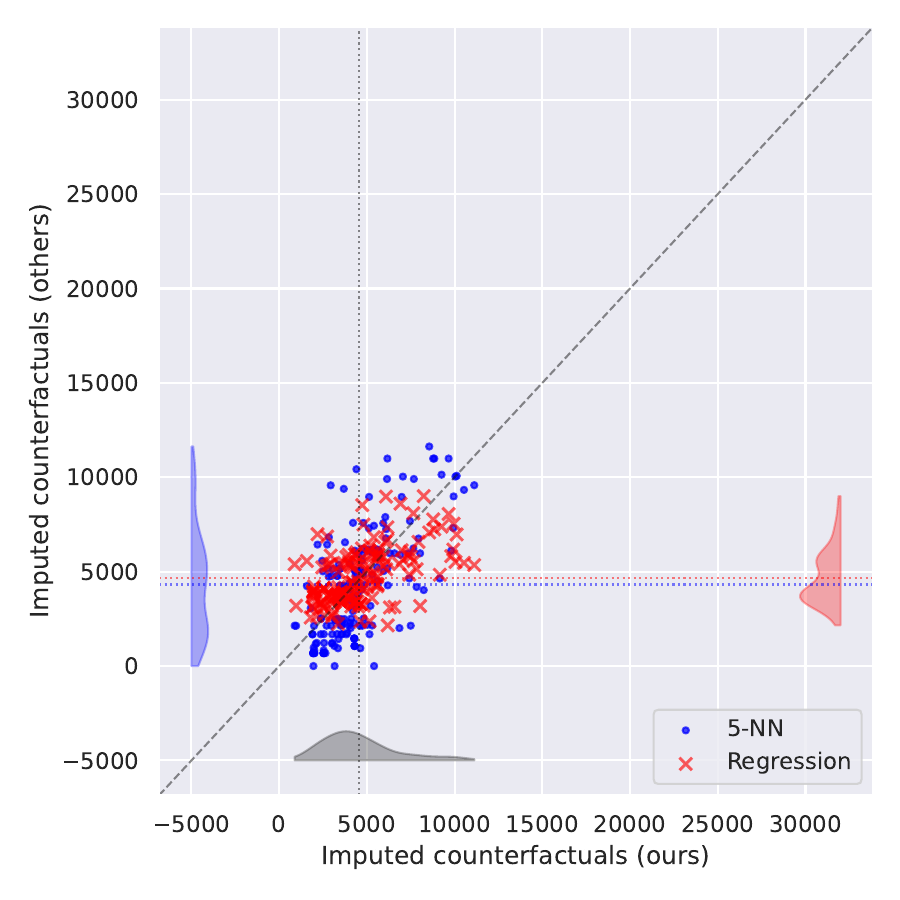}}
	\subfloat[NSW $\lambda = 0.01$]{\includegraphics*[width=0.47\textwidth]{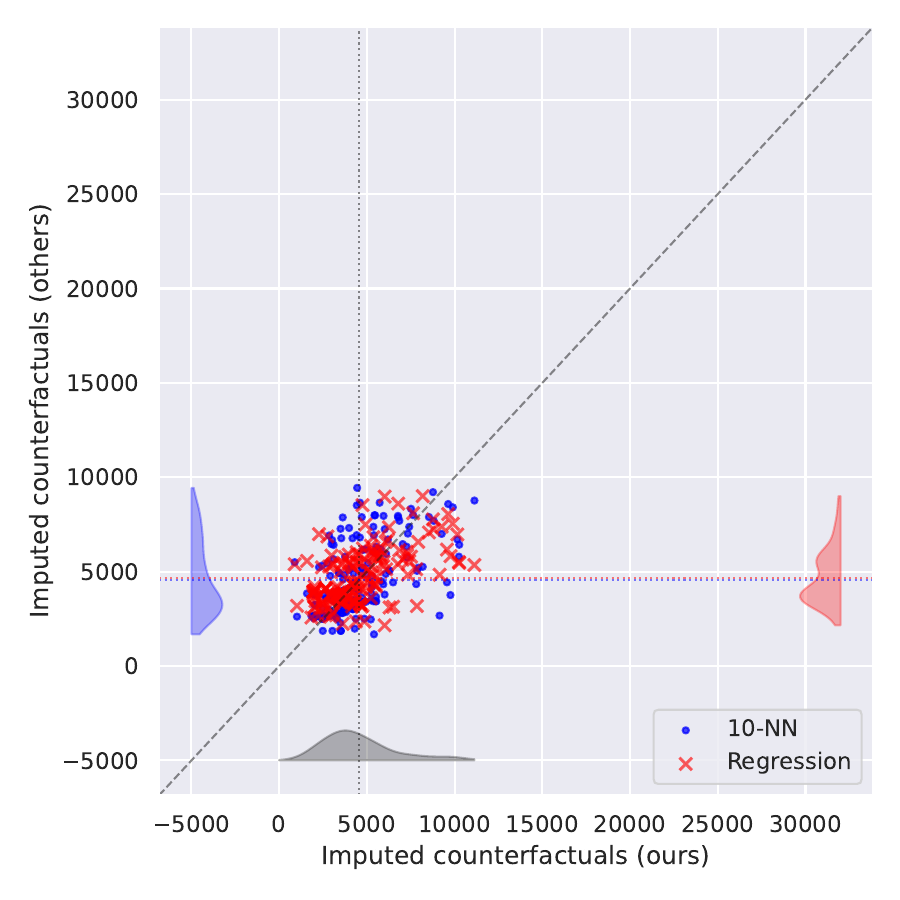}}
	\qquad
	\subfloat[PSID $\lambda = 0.001$]{\includegraphics*[width=0.47\textwidth]{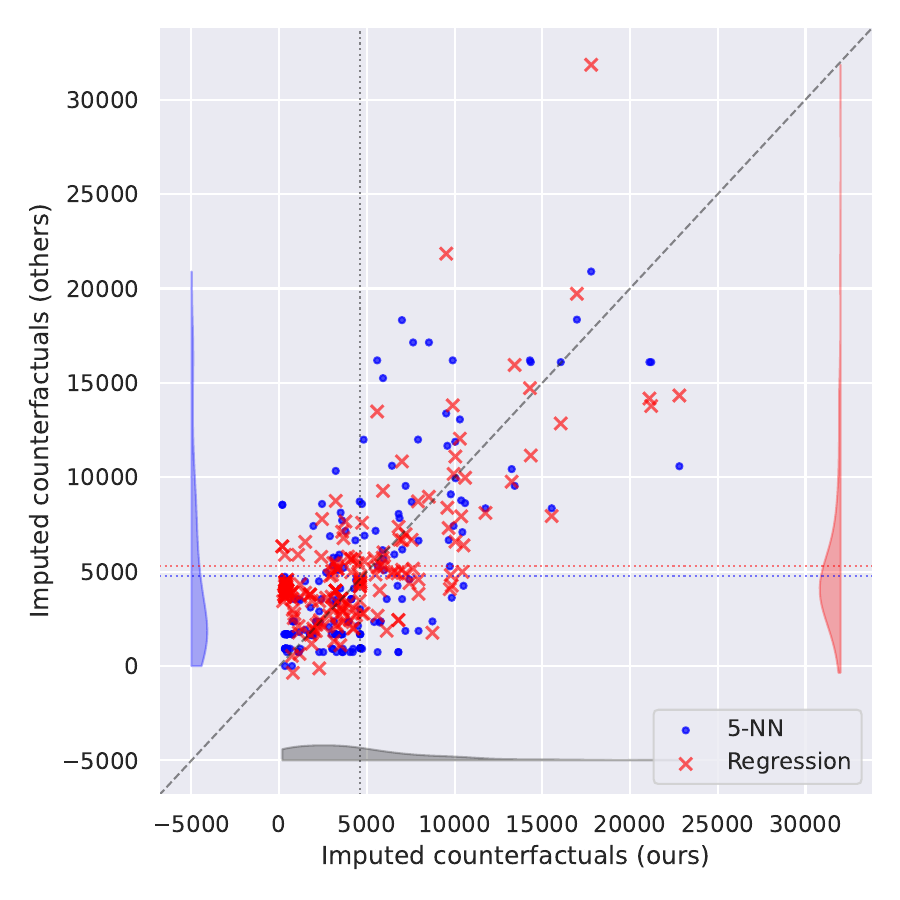}}
	\subfloat[PSID $\lambda = 0.01$]{\includegraphics*[width=0.47\textwidth]{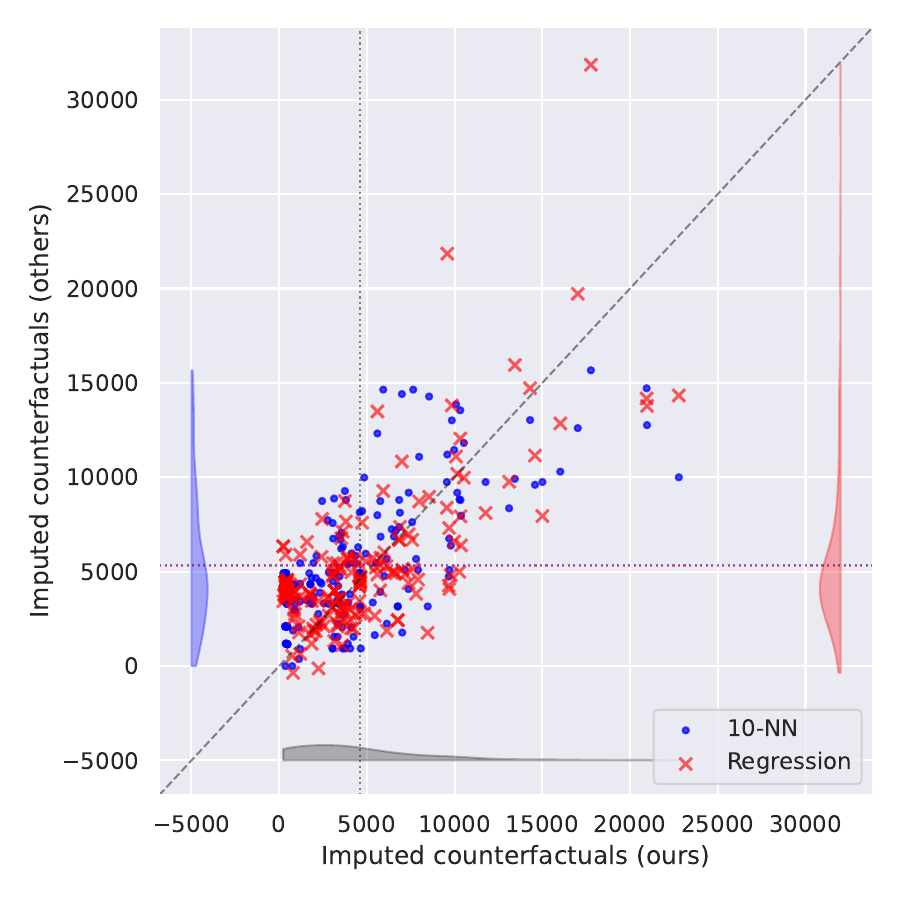}}
	\caption{\small Scatter plots of the imputed counterfactual outcomes of the treated along with the histograms of the marginal distributions. The setup is the same as in Figure \ref{fig:NSW_imputation_scatterplots}, but the results are based on the polynomial kernel of degree two for both NSW and PSID data.}
	\label{fig:NSW_imputation_scatterplots_poly}
\end{figure}

\subsection{Inference}
We apply the method presented in Section \ref{sec-inference} to construct confidence intervals around the imputed counterfactual outcomes. Figure \ref{fig:NSW_ITEvs75} (a) shows the results for the experimental data using the linear kernel with $\lambda \in \{0.001, 1\}$, which plots the estimated individual treatment effects (ITEs), $\widehat{\tau}_j = Y_j - \widehat{Y}_j(0)$, against the logarithm of earnings in 1975, where the confidence intervals are shown as error bars around the estimated ITEs. Here, only 74 treated units (out of 185 treated units) with positive earnings in 1975 are shown. We pick this index merely for visualization of individualization. For comparison, we also show the blue points representing the ITEs based on imputing all the treated units with the mean of the control group outcomes, along with the confidence intervals based on the standard error of the mean; the case where there is no individualization in the imputed control outcomes, and thus no meaningful individual confidence intervals addressing the bias. Let us first focus on the point estimates of the ITEs. We can see that most of the estimated ITEs fall in the interval $(-10000, 10000)$, seemingly uncorrelated with the earnings in 1975, while there are seven individuals whose ITEs are above 10000. These seven units with the largest ITEs are likely to be the ones who benefit most from the treatment, standing out from the rest of the treatment group. The confidence intervals help us decide how sure we are about these seemingly significant ITEs. We can confidently say that the top two ITEs are extremely large compared to the rest of the ITEs since the lower ends of their individual confidence intervals are larger than the upper ends of the rest of the ITEs, regardless of the choice of $\lambda$. Now, looking at the confidence intervals, rather than only looking at the point estimates, the four units whose ITEs are between 10000 and 30000 are less distinguishable from the rest. The error bars do help to suggest which units benefit the more here. For the third largest ITE, around 20000, the conclusion may vary depending on $\lambda$. If we read the confidence intervals corresponding to $\lambda = 1$ (red, dashed), this unit seems to be more significant than the rest, while this conclusion is less assertive for $\lambda = 0.001$ as the confidence interval of this unit overlaps with the interval $(-15000, 20000)$ which contains the rest of the confidence intervals. In other words, using $\lambda = 0.001$ leads to the most conservative conclusion regarding who benefits the most from the treatment as the confidence intervals are wider. Unlike the confidence intervals obtained by the proposed method, the intervals based on the mean imputation (shown in blue) are extremely narrow, which does not take into account individual heterogeneity, overlooking a potentially large bias; the confidence interval here is useless for individual decision making and may result in overly optimistic conclusions, for instance, essentially the majority benefits significantly from the program. (b) and (c) of Figure \ref{fig:NSW_ITEvs75} repeat the same analysis using the RBF kernel and the polynomial kernel of degree two, respectively. The overall patterns are similar to those using the linear kernel, but the confidence intervals tend to be wider. Accordingly, even the second largest ITE above 30000 is not as distinguishable from the rest as in the linear kernel case. 

The results based on the PSID data are shown in Figure \ref{fig:PSIDtrim-IPW_ITEvs75}. One notable difference is that the intervals are generally much wider than those of the NSW experimental data. One possible explanation is that the PSID data, even after trimmed, differs significantly from the experimental data, leading to the matching with more conservative error estimates. Accordingly, the confidence intervals are wider, which makes it harder to make a decisive conclusion with confidence.

Lastly, we briefly comment on the choice of $\rho$, the regularization parameter of the kernel ridge regression mentioned in Section \ref{sec:confidence_intervals_method}. Increasing $\rho$ leads to a smaller estimate of $\|f_0\|_\cH$, leading to a smaller bias correction. Increasing $\rho$ will typically result in a larger estimate of $\sigma_0$, implying a large variability. The correct choice of $\rho$ depends on the underlying signal-to-noise ratio for the data-generating process in the control outcomes, namely, the complexity of the function vs.\ the amount of noise. In all examples of this section, we choose $\rho$ using 5-fold cross-validation. 

\begin{figure}
	\centering
	\subfloat[Linear]{\includegraphics*[width=0.99\textwidth]{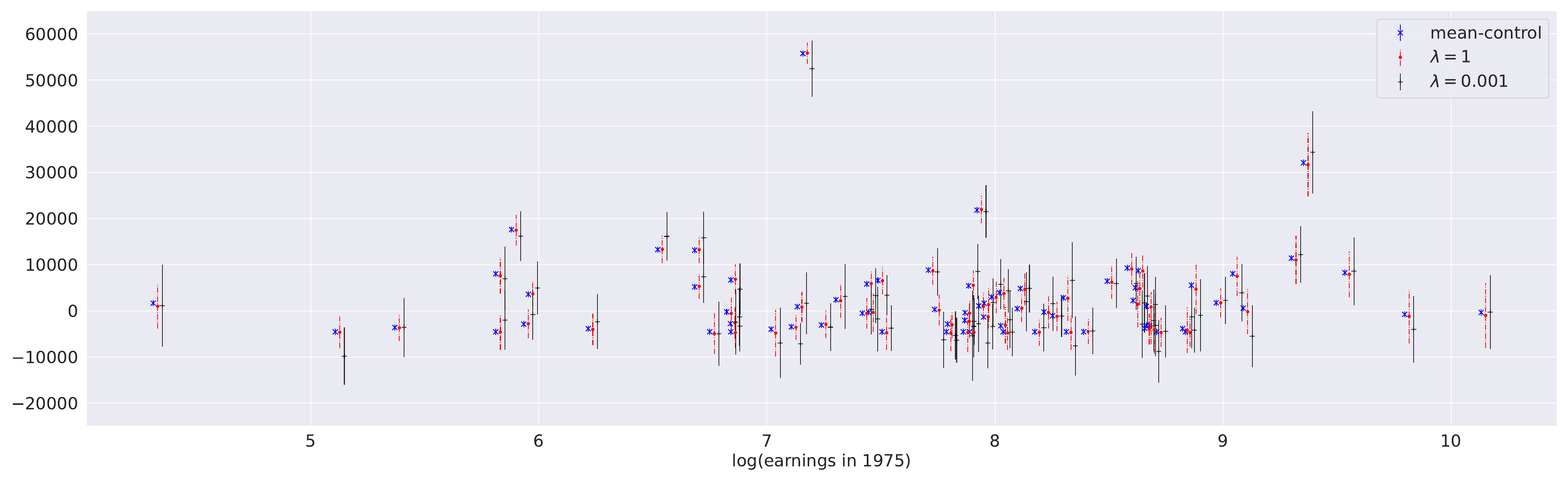}}
	\qquad
	\subfloat[RBF]{\includegraphics*[width=0.99\textwidth]{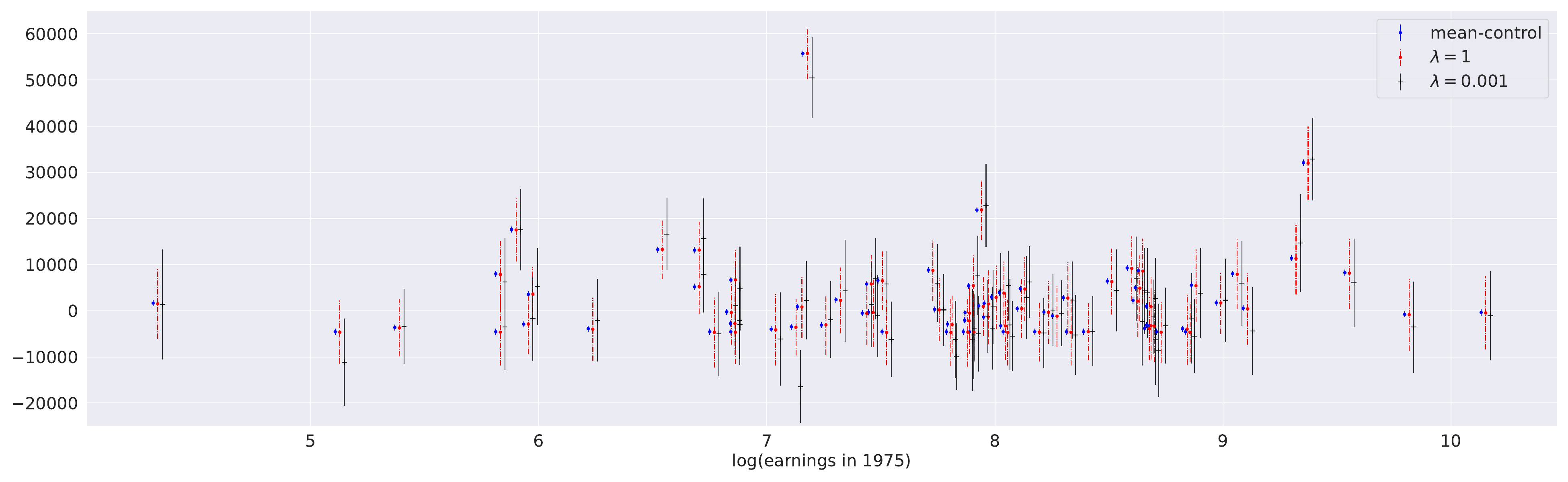}}
	\qquad
	\subfloat[Polynomial]{\includegraphics*[width=0.99\textwidth]{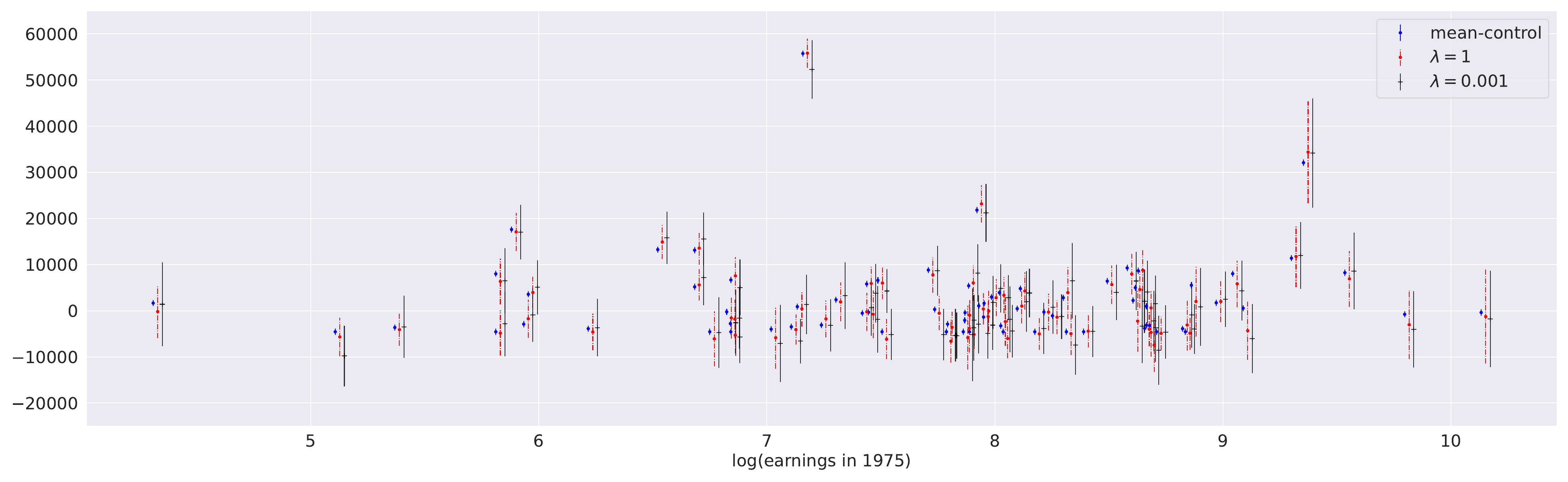}}
	\caption{\small NSW experimental data: estimated individual treatment effects (ITEs) of the treated together with the confidence intervals along the earnings before the treatment. Confidence intervals are shown as error bars around the estimated ITEs. 74 treated units (out of 185 treated units) with positive earnings in 1975 are shown. The $x$-axis is the logarithm of earnings in 1975, while the $y$-axis is the estimated ITEs. The results are based on the experimental data using uniform weights $v, w$ in \eqref{eq:kernelized_formulation} with three different kernels: linear, RBF, and polynomial of degree two. The blue points correspond to imputing all the treated units with the mean of the control group outcomes, where the error bars show the confidence interval based on the standard error of the mean.}
	\label{fig:NSW_ITEvs75}
\end{figure}

\begin{figure}
	\centering
	\subfloat[Linear]{\includegraphics*[width=0.99\textwidth]{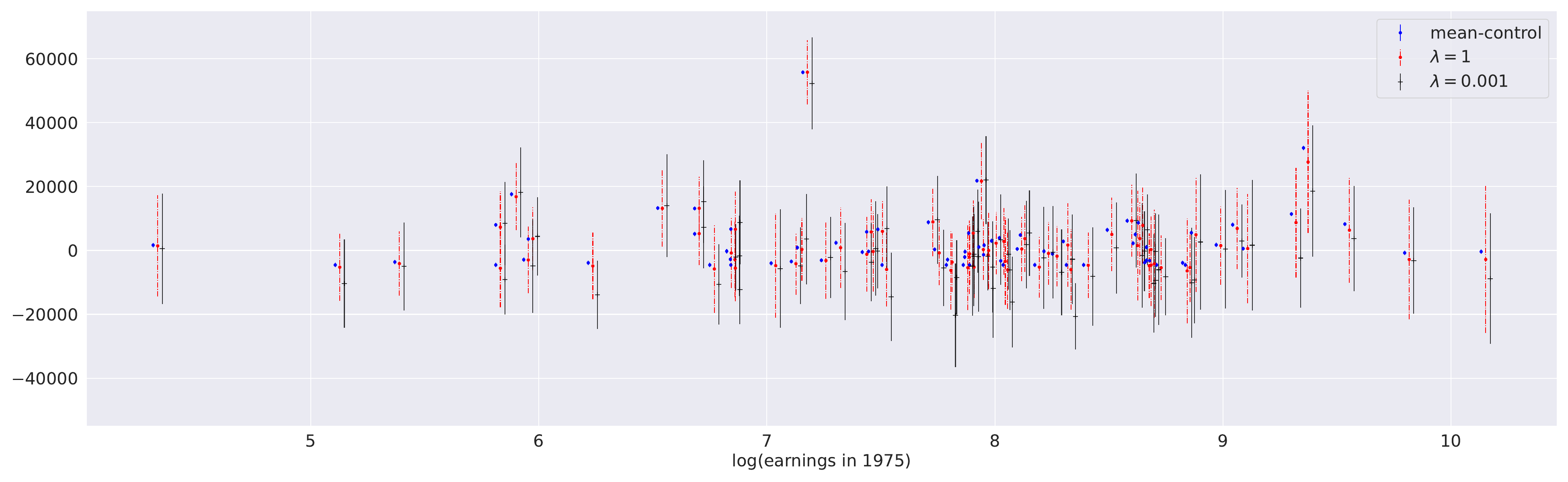}}
	\qquad
	\subfloat[RBF]{\includegraphics*[width=0.99\textwidth]{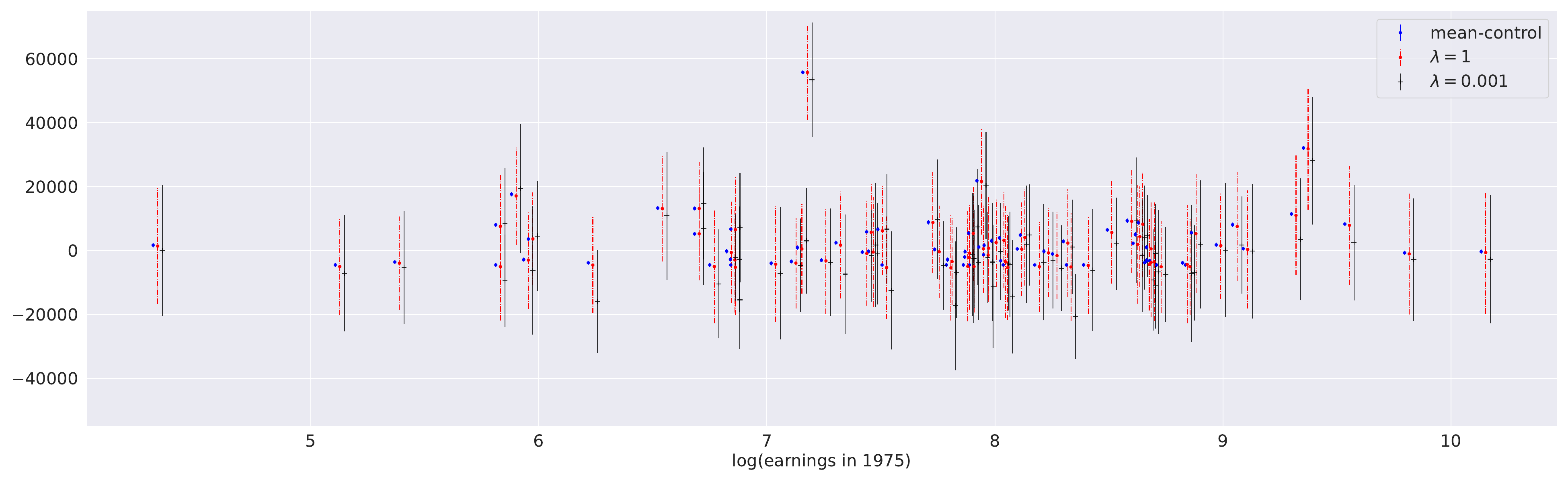}}
	\qquad
	\subfloat[Polynomial]{\includegraphics*[width=0.99\textwidth]{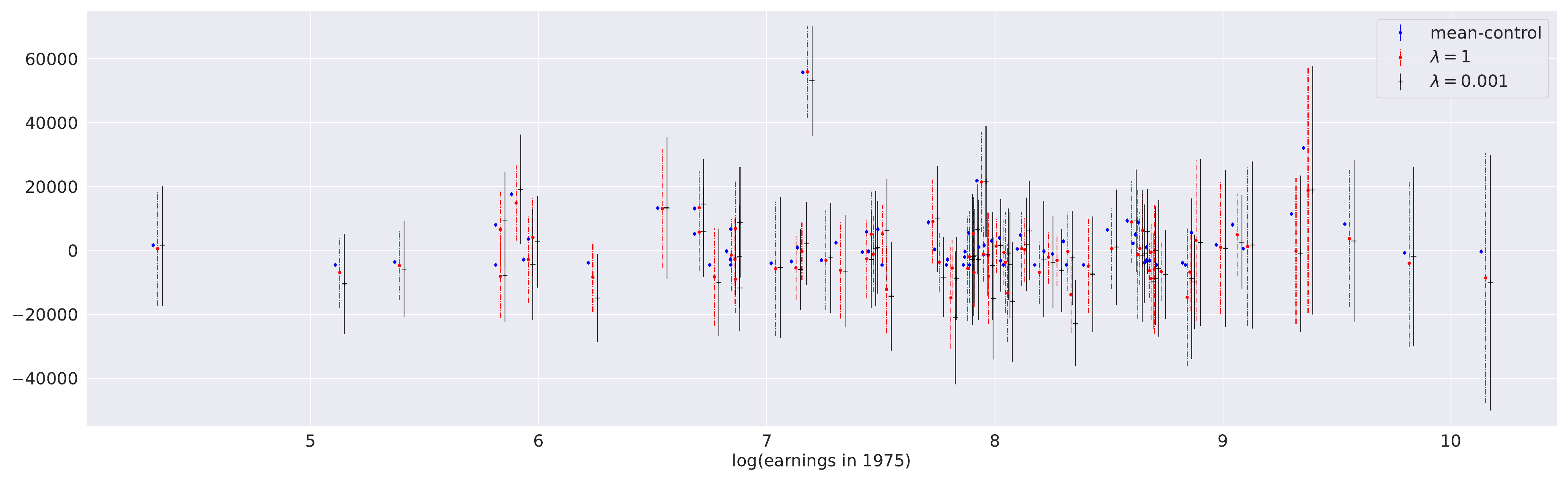}}
	\caption{\small PSID nonexperimental data: estimated individual treatment effects (ITEs) of the treated along the earnings before the treatment. The setup is the same as in Figure \ref{fig:NSW_ITEvs75}, but the results are based on the trimmed PSID data. For the weights $v, w$ in \eqref{eq:kernelized_formulation}, we let $v$ be uniform, while $w$ is the propensity score-based weights as in \eqref{eq:propensity_score_weights}.}
	\label{fig:PSIDtrim-IPW_ITEvs75}
\end{figure}

\section{Conclusion}
We proposed a convexified matching method for missing value imputation and individualized inference, integrating favorable features from optimal matching, regression imputation, and synthetic control. We impute counterfactual outcomes based on convex combinations of observed outcomes, defined by an optimal coupling between the treated and control data sets. Finding an optimal coupling is a convex relaxation to the combinatorial optimal matching problem, for which we propose efficient algorithms based on matrix scaling. Unlike existing imputation methods, we begin with a desirable aggregate-level summary and estimate granular-level individual treatment effects by properly constraining the coupling so that the estimated individual effects are consistent with the aggregate summary. We provided a method to construct individual confidence intervals for the estimated counterfactual outcomes, along with a simulation study to demonstrate the effectiveness of our method. The simulation confirms our theoretical insight on the level of entropic regularization needed. We establish that entropic regularization plays a crucial role in both efficiency in inference and computation: first, the convexified matching objective is gauged for minimizing the width of the individual confidence intervals, trading off bias and variance; second, the entropic regularization enables us to design fast algorithms. We demonstrated the empirical performance of our method on the NSW data, using both experimental and nonexperimental data sets, with a comparison to existing methods and robustness checks.

\printbibliography

\newpage 
\appendix
\section{Proofs}
\subsection{Proof of Proposition \ref{prop:Hellinger-KL-Chisq}}

\begin{proof}[Proof of Proposition \ref{prop:Hellinger-KL-Chisq}]
	For any $p = (p_1, \ldots, p_n) \in \Delta_n$, we have
	\begin{equation*}
		\sum_{i = 1}^{n} p_i \log(p_i) + \log(n) \le n \sum_{i = 1}^{n} p_i^2 - 1,
	\end{equation*}
	where the left-hand side is the Kullback-Leibler (KL) divergence of $p$ from $\frac{1}{n} 1_n$ which is known to be bounded by the $\chi^2$ distance of $p$ from $\frac{1}{n} 1_n$ on the right-hand side; see Lemma 2.7 of \cite{tsybakov_2009}. Meanwhile, the KL divergence is bounded below by the Hellinger distance; see Lemma 2.4 of \cite{tsybakov_2009}. Therefore, we have
	\begin{equation*}
		\sum_{i = 1}^{n} p_i \log(p_i) + \log(n) \ge \sum_{i = 1}^{n} (\sqrt{p_i} - \sqrt{1 / n})^2,
	\end{equation*}
	where the right-hand side is the Hellinger distance between $p$ and $\frac{1}{n} 1_n$. The right-hand side can be further lower bounded,
	\begin{equation*} 
		\sum_{i = 1}^{n} (\sqrt{p_i} - \sqrt{1 / n})^2
		= \sum_{i = 1}^{n} \frac{(p_i - 1 / n)^2}{(\sqrt{p_i} + \sqrt{1 / n})^2}
		\ge \frac{n \sum_{i = 1}^{n} p_i^2 - 1}{2 (n \|p\|_\infty + 1)}.	
	\end{equation*}
	Combining these inequalities, we have 
	\begin{equation*}
		\frac{1}{2(n \|p\|_\infty + 1)} \le \frac{\sum_{i = 1}^{n} p_i \log(p_i) + \log(n)}{n \sum_{i = 1}^{n} p_i^2 - 1} \le 1.
	\end{equation*}
\end{proof}

\subsection{Proof of Proposition \ref{prop:minimizer=fixed_point}}
\begin{proof}[Proof of Proposition \ref{prop:minimizer=fixed_point}]
	We have already discussed that \eqref{eq:optim} must admit a unique minimizer $\pi^\star \in \Pi_{a, b}^{+}$. We show that $\pi^\star$ is a unique fixed point of $T_\lambda$, which follows from the Karush-Kuhn-Tucker (KKT) conditions. To see this, rewrite \eqref{eq:optim} as 
	\begin{equation*}
		\begin{aligned}
			\min_{\pi \in \R_+^{n \times m}} \quad & (g(\pi) + \lambda h(\pi)) \\
			\mathrm{subject~to} \quad & \pi 1_m = a \quad \text{and} \quad \pi^\top 1_n = b.
		\end{aligned}
	\end{equation*}
	Then, define the Lagrangian $L \colon \R_+^{n \times m} \times \R^n \times \R^m \to \R$ by
	\begin{equation*}
		L(\pi, \mu, \nu) = g(\pi) + \lambda h(\pi) + \langle \mu, \pi 1_m - a \rangle + \langle \nu, \pi^\top 1_n - b \rangle.
	\end{equation*}
	The KKT conditions for $(\pi, \mu, \nu) \in \R_+^{n \times m} \times \R^n \times \R^m$ are $\pi \in \Pi_{a, b}$ and $\pi$ is a minimizer of $L(\cdot, \mu, \nu)$ over $\R_+^{n \times m}$. Due to $h$, we can see that $L(\cdot, \mu, \nu)$ must admit a unique minimizer at the interior of $\R_+^{n \times m}$, namely, all the entries of the unique minimizer are strictly positive. Accordingly, the minimizer must satisfy the first-order condition: 
	\begin{equation*}
		\nabla_\pi L(\pi, \mu, \nu) = \nabla g(\pi) + \lambda \nabla h(\pi) + \mu 1_m^\top + 1_n \nu^\top = 0,
	\end{equation*}
	which leads to 
	\begin{equation}
		\label{eq:lagrangian_optimality}
		\pi = \exp\left(-\frac{\mu 1_m^\top + 1_n \nu^\top + \nabla g(\pi)}{\lambda}\right).
	\end{equation}
	In summary, the KKT conditions are $\pi \in \Pi_{a, b}$ and \eqref{eq:lagrangian_optimality}, which---by Definition \ref{prop:EOT_Sinkhorn}---leads to the following:
	\begin{equation*}
		\pi = \Phi_\lambda(\nabla g(\pi)) = T_\lambda(\pi).
	\end{equation*}
	Hence, we have shown that $\pi$ is a minimizer of \eqref{eq:optim} if and only if $\pi$ is a fixed point of $T_\lambda$. As we have already proved that \eqref{eq:optim} admits a unique minimizer $\pi^\star$, we conclude that $T_\lambda$ has a unique fixed point $\pi^\star$.
\end{proof}

\subsection{Proof of Theorem \ref{thm:fixed-point_L1}}
Theorem \ref{thm:fixed-point_L1} is based on the following lemma that analyzes the Lipschitz property of the operator $\Phi_\lambda$---defined in Proposition \ref{prop:EOT_Sinkhorn}---w.r.t. the $L^\infty$ metric. 
\begin{lemma}
	\label{lem:EOT_perturbation}
	Fix $\lambda > 0$. Define a map $\Phi_\lambda \colon \R^{n \times m} \to \Pi_{a, b}^{+}$ by 
	\begin{equation*}
		\Phi_\lambda(C) = \argmin_{\pi \in \Pi_{a, b}} \left(\langle C, \pi \rangle + \lambda h(\pi)\right).
	\end{equation*}
	Then, for any $C_1, C_2 \in \R^{n \times m}$, we have
	\begin{equation}
		\label{eq:EOT_perturbation_L1}
		\|\Phi_\lambda(C_1) - \Phi_\lambda(C_2)\|_1 \le \frac{\|C_1 - C_2\|_\infty}{\lambda}.
	\end{equation}
\end{lemma}

\begin{proof}[Proof of Lemma \ref{lem:EOT_perturbation}]
	We first claim that
	\begin{equation}
		\label{eq:entropic_ot_optimality_direction}
		\langle C_1 + \lambda \nabla h(\Phi_\lambda(C_1)), \pi - \Phi_\lambda(C_1)\rangle \ge 0 \quad \forall \pi \in \Pi_{a, b}.
	\end{equation}
	Suppose not, namely, there exists $\pi \in \Pi_{a, b}$ such that 
	\begin{equation*}
		\langle C_1 + \lambda \nabla h(\Phi_\lambda(C_1)), \pi - \Phi_\lambda(C_1)\rangle < 0.
	\end{equation*}
	In other words, letting $q(\pi) := \langle C_1, \pi \rangle + \lambda h(\pi)$, we have $\pi \in \Pi_{a, b}$ such that 
	\begin{equation*}
		\langle \nabla q(\Phi_\lambda(C_1)), \pi - \Phi_\lambda(C_1) \rangle < 0.
	\end{equation*} 
	As $q$ is differentiable at $\Phi_\lambda(C_1)$ and $q$ is convex on $\Pi_{a, b}$, this means that we can find a point $\pi'$ on the line segment between $\Phi_\lambda(C_1)$ and $\pi$ such that $q(\pi') < q(\Phi_\lambda(C_1))$, which contradicts that $\Phi_\lambda(C_1)$ is a minimizer of $q$ over $\Pi_{a, b}$. Hence, \eqref{eq:entropic_ot_optimality_direction} must hold. Letting $\pi = \Phi_\lambda(C_2)$ in \eqref{eq:entropic_ot_optimality_direction}, we have 
	\begin{equation*}
		\langle C_1 + \lambda \nabla h(\Phi_\lambda(C_1)), \Phi_\lambda(C_2) - \Phi_\lambda(C_1)\rangle \ge 0.
	\end{equation*}
	Changing the role of $C_1$ and $C_2$, 
	\begin{equation*}
		\langle C_2 + \lambda \nabla h(\Phi_\lambda(C_2)), \Phi_\lambda(C_1) - \Phi_\lambda(C_2)\rangle \ge 0.
	\end{equation*}
	Combining the above two inequalities, we have 
	\begin{equation}
		\label{eq:EOT_perturbation_rough}
		\langle \nabla h(\Phi_\lambda(C_1)) - \nabla h(\Phi_\lambda(C_2)), \Phi_\lambda(C_1) - \Phi_\lambda(C_2) \rangle 
		\le 
		- \frac{1}{\lambda} \langle C_1 - C_2, \Phi_\lambda(C_1) - \Phi_\lambda(C_2) \rangle.
	\end{equation}
	By applying H{\"o}lder's inequality to the right-hand side of \eqref{eq:EOT_perturbation_rough}, we have 
	\begin{equation}
		\label{eq:EOT_rough_Holder}
		- \frac{1}{\lambda} \langle C_1 - C_2, \Phi_\lambda(C_1) - \Phi_\lambda(C_2) \rangle \le \frac{\|C_1 - C_2\|_\infty \cdot \|\Phi_\lambda(C_1) - \Phi_\lambda(C_2)\|_1}{\lambda}.
	\end{equation}
	Also, as $h$ is $1$-strongly convex with respect to $\|\cdot\|_1$ on $\Delta_{n, m}$, we have 
	\begin{equation*}
		\langle \nabla h(P_1) - \nabla h(P_2), P_1 - P_2 \rangle \ge \|P_1 - P_2\|_1^2 \quad \forall P_1, P_2 \in \Delta_{n, m},
	\end{equation*}
	which allows us to lower bound the left-hand side of \eqref{eq:EOT_perturbation_rough} as follows:
	\begin{equation}
		\label{eq:EOT_rough_strong_cvx_L1}
		\langle \nabla h(\Phi_\lambda(C_1)) - \nabla h(\Phi_\lambda(C_2)), \Phi_\lambda(C_1) - \Phi_\lambda(C_2) \rangle \ge \|\Phi_\lambda(C_1) - \Phi_\lambda(C_2)\|_1^2.
	\end{equation}
	Combining \eqref{eq:EOT_perturbation_rough}, \eqref{eq:EOT_rough_Holder}, and \eqref{eq:EOT_rough_strong_cvx_L1}, we have \eqref{eq:EOT_perturbation_L1}.
\end{proof}

\begin{proof}[Proof of Theorem \ref{thm:fixed-point_L1}]
	By \eqref{eq:EOT_perturbation_L1} of Lemma \ref{lem:EOT_perturbation}, for any $A_1, A_2 \in \R^{n \times m}$.
	\begin{equation*}
		\begin{split}
			\|T_\lambda(A_1) - T_\lambda(A_2)\|_1 
			& = \|\Phi_\lambda(\nabla g(A_1)) - \Phi_\lambda(\nabla g(A_2))\|_1 \\
			& \le \frac{\|\nabla g(A_1) - \nabla g(A_2)\|_\infty}{\lambda} \\
			& = \frac{\|H (A_1 - A_2)\|_\infty}{\lambda} \\
			& \le \frac{\|H (A_1 - A_2)\|_\infty}{\lambda} \\
			& \le \frac{\|H\|_\infty}{\lambda} \|A_1 - A_2\|_1,			
		\end{split}
	\end{equation*}
	where the last inequality follows from $\|H (A_1 - A_2)\|_\infty \le \|H\|_\infty \|A_1 - A_2\|_1$. Therefore, $T_\lambda$ is a contraction on $\Pi_{a, b}$ equipped with $\|\cdot\|_1$ if $\lambda > L = \|H\|_\infty$. By the Banach–Caccioppoli theorem, $T_\lambda$ has a unique fixed point, and we have \eqref{eq:fixed_point_bound_L1}.
\end{proof}

\subsection{Proof of Theorem \ref{thm:steepest-descent-KL-convergence}}
\begin{proof}[Proof of Theorem \ref{thm:steepest-descent-KL-convergence}]
	We first show the following from \eqref{eqn-steepest-descent}: for any $k \ge 0$ and $\pi \in \Pi_{a, b}$,
	\begin{equation}
		\label{eq:steepest-descent-KL-update-optimality}
		\langle \nabla f(\pi^{(k)}), \pi^{(k + 1)} - \pi \rangle \le \frac{\langle \nabla h(\pi^{(k)}) - \nabla h(\pi^{(k + 1)}), \pi^{(k + 1)} - \pi \rangle}{\tau_k}.
	\end{equation}
	To see this, pick $\pi \in \Pi_{a, b}$ and define $q(\epsilon) = f((1 - \epsilon) \pi^{(k + 1)} + \epsilon \pi) + \tau_k^{-1} D_h((1 - \epsilon) \pi^{(k + 1)} + \epsilon \pi, \pi^{(k)})$ for $\epsilon \in [0, 1]$. As $q(\epsilon) \ge q(0)$ for $\epsilon \in [0, 1]$, the right derivative of $q$ at $\epsilon = 0$ must be nonnegative, which leads to \eqref{eq:steepest-descent-KL-update-optimality}. Next, letting $D_f$ be the Bregman divergence with respect to $f$, we have
	\begin{equation*}
		\begin{split}
			f(\pi^{(k + 1)}) - f(\pi) &= f(\pi^{(k + 1)}) - f(\pi^{(k)}) +  f(\pi^{(k)}) -  f(\pi)\\
			& = D_f(\pi^{(k + 1)}, \pi^{(k)}) + \langle \nabla f(\pi^{(k)}), \pi^{(k + 1)} - \pi^{(k)} \rangle + f(\pi^{(k)}) - f(\pi) \\
			& \le D_f(\pi^{(k + 1)}, \pi^{(k)}) + \langle \nabla f(\pi^{(k)}), \pi^{(k + 1)} - \pi^{(k)} \rangle + \langle \nabla f(\pi^{(k)}), \pi^{(k)} - \pi \rangle \\
			& \le D_f(\pi^{(k + 1)}, \pi^{(k)}) + \frac{\langle \nabla h(\pi^{(k)}) - \nabla h(\pi^{(k + 1)}), \pi^{(k + 1)} - \pi \rangle}{\tau_k}, 
		\end{split}
	\end{equation*}
	where the first inequality uses the convexity of $f$ and the second inequality follows from \eqref{eq:steepest-descent-KL-update-optimality}. For the last term on the right-hand side, we have
	\begin{equation*}
		\begin{split}
			\langle \nabla h(\pi^{(k)}) - \nabla h(\pi^{(k + 1)}), \pi^{(k + 1)} - \pi \rangle 
			& = \langle \log \frac{\pi^{(k)}}{\pi^{(k + 1)}}, \pi^{(k + 1)} - \pi \rangle \\
			& = \langle \log \frac{\pi}{\pi^{(k)}}, \pi \rangle - \langle \log \frac{\pi}{\pi^{(k + 1)}}, \pi \rangle - \langle \log \frac{\pi^{(k + 1)}}{\pi^{(k)}}, \pi^{(k + 1)} \rangle \\
			& = D_h(\pi, \pi^{(k)}) - D_h(\pi, \pi^{(k + 1)}) - D_h(\pi^{(k + 1)}, \pi^{(k)}).
		\end{split}
	\end{equation*}
	Hence, 
	\begin{equation}
		\label{eq:gap_inequality}
		f(\pi^{(k + 1)}) - f(\pi) \le D_f(\pi^{(k + 1)}, \pi^{(k)}) + \frac{D_h(\pi, \pi^{(k)}) - D_h(\pi, \pi^{(k + 1)}) - D_h(\pi^{(k + 1)}, \pi^{(k)})}{\tau_k}.
	\end{equation}
	Meanwhile, 
	\begin{equation*}
		\begin{split}
			D_f(\pi^{(k + 1)}, \pi^{(k)})
			& = \frac{1}{2} \langle \pi^{(k + 1)} - \pi^{(k)}, H (\pi^{(k + 1)} - \pi^{(k)}) \rangle + \lambda D_h(\pi^{(k + 1)}, \pi^{(k)}) \\
			& \le \frac{\|H\|_\infty}{2} \|\pi^{(k + 1)} - \pi^{(k)}\|_1^2 + \lambda D_h(\pi^{(k + 1)}, \pi^{(k)}) \\
			& \le (\|H\|_\infty + \lambda) D_h(\pi^{(k + 1)}, \pi^{(k)}),
		\end{split}
	\end{equation*}
	where the first inequality follows from H{\"o}lder's inequality and the second inequality uses Pinsker's inequality $\|\pi^{(k + 1)} - \pi^{(k)}\|_1^2 \le 2 D_h(\pi^{(k + 1)}, \pi^{(k)})$; see, Lemma 2.5 of \cite{tsybakov_2009}. Now, as $\tau_k^{-1} = \tau^{-1} \ge \|H\|_\infty + \lambda$, we have
	\begin{equation*}
		D_f(\pi^{(k + 1)}, \pi^{(k)}) \le \frac{D_h(\pi^{(k + 1)}, \pi^{(k)})}{\tau},
	\end{equation*}
	which, together with \eqref{eq:gap_inequality}, leads to 
	\begin{equation*}
		f(\pi^{(k + 1)}) - f(\pi) \le \frac{D_h(\pi, \pi^{(k)}) - D_h(\pi, \pi^{(k + 1)})}{\tau} \quad \forall \pi \in \Pi_{a, b}.
	\end{equation*}
	By letting $\pi = \pi^{(k)}$, we have $f(\pi^{(k + 1)}) \le f(\pi^{(k)})$ for all $k \ge 0$. Hence, 
	\begin{equation*}
		f(\pi^{(T)}) - f(\pi) \le \frac{1}{T} \sum_{k = 0}^{T - 1} (f(\pi^{(k + 1)}) - f(\pi)) \le \frac{1}{T} \frac{D_h(\pi, \pi^{(0)}) - D_h(\pi, \pi^{(T)})}{\tau} \le \frac{1}{T} \frac{D_h(\pi, \pi^{(0)})}{\tau}.
	\end{equation*}
	As $f$ always admits a minimizer on $\Pi_{a, b}$ (even when $\lambda = 0$), plugging a minimizer to $\pi$ and using the fact that $D_h(\pi, \pi^{(0)}) \le \log(n m)$, we obtain \eqref{eq:steepest-descent-KL-convergence}. Lastly, 
	\begin{equation*}
		\begin{split}
			f(\pi^{(T)}) - f(\pi^\star) 
			& = \langle \nabla f(\pi^\star), \pi^{(T)} - \pi^\star \rangle + D_f(\pi^{(T)}, \pi^\star) \\
			& = \langle \nabla f(\pi^\star), \pi^{(T)} - \pi^\star \rangle + \frac{1}{2} \langle \pi^{(T)} - \pi^{(k)}, H (\pi^{(T)} - \pi^{(k)}) \rangle + \lambda D_h(\pi^{(T)}, \pi^\star) \\
			& \ge \frac{\lambda}{2} \|\pi^{(T)} - \pi^\star\|_1^2,
		\end{split}
	\end{equation*}
	where the last inequality uses Pinsker's inequality and $\langle \nabla f(\pi^\star), \pi^{(T)} - \pi^\star \rangle \ge 0$, which holds as $\pi^\star$ is a minimizer of $f$. Hence, we have \eqref{eq:convergence_to_minimizer}.
\end{proof}

\end{document}